\documentclass{article}
\usepackage{graphicx} 

\usepackage{authblk}

\usepackage{float}

\usepackage[dvipsnames]{xcolor}
\usepackage{xcolor,todonotes}
\usepackage{enumitem}

\usepackage{tikz}
\tikzstyle{edge}=[shorten <=2pt, shorten >=2pt, >=stealth, line width=1.5pt]
\tikzstyle{vertex}=[circle, fill=white, draw, minimum size=6pt, inner sep=0pt,
                    outer sep=0pt]
\tikzstyle{arc}=[->, shorten <=3pt, shorten >=3pt, >=stealth, line width=1.25pt]

\usepackage{amsmath}
\usepackage{amssymb,hyperref}
\usepackage{bm} 

\usepackage{amsthm,geometry,cite}

\usepackage{multicol}

\newtheorem{theorem}{Theorem}
\newtheorem{definition}{Definition}
\newtheorem{conjecture}{Conjecture}
\newtheorem{lemma}[theorem]{Lemma}
\newtheorem{corollary}[theorem]{Corollary}
\newtheorem{proposition}[theorem]{Proposition}
\newtheorem{observation}[theorem]{Observation}

\theoremstyle{remark}
\newtheorem{remark}{Remark}

\theoremstyle{definition}
\newtheorem{example}{Example}

\DeclareMathOperator{\Pol}{Pol}

\DeclareMathOperator{\Betw}{Betw}
\DeclareMathOperator{\Age}{Age}
\DeclareMathOperator{\Bip}{Bip}
\DeclareMathOperator{\Aut}{Aut}

\newcommand{\calF}{\mathcal F}
\newcommand{\calC}{\mathcal C}
\DeclareMathOperator{\CSP}{CSP}
\DeclareMathOperator{\injCSP}{injCSP}
\DeclareMathOperator{\NP}{NP}
\DeclareMathOperator{\PO}{P}
\DeclareMathOperator{\coNP}{coNP}

\DeclareMathOperator{\PCSP}{PCSP}
\DeclareMathOperator{\SP}{SP}
\DeclareMathOperator{\GR}{Grid}

\title{A CSP approach to Graph Sandwich Problems}

\author[1]{Manuel Bodirsky\thanks{manuel.bodirsky@tu-dresden.de}}
\author[1]{Santiago Guzm\'an-Pro\thanks{santiago.guzman\_pro@tu-dresden.de\\
This project has been funded by the European Research Council (Project POCOCOP, ERC Synergy Grant 101071674).
Views and opinions expressed are however those of the authors only and do not necessarily reflect those of
the European Union or the European Research Council Executive Agency. Neither the European Union nor the
granting authority can be held responsible for them.}}

\affil[1]{Institut f\"ur Algebra, TU Dresden}
\date{\today}

\begin{document}

\maketitle

\begin{abstract}
  The \emph{Sandwich Problem} (SP) for a graph class $\calC$ is the following computational problem. 
  The input is a pair of graphs $(V,E_1)$ and $(V,E_2)$ where $E_1\subseteq E_2$, and the task
  is to decide whether there is an edge set $E$ where $E_1\subseteq E\subseteq E_2$ such that the graph
  $(V,E)$ belongs to $\calC$. In this paper we show that many SPs correspond to the
  constraint satisfaction problem (CSP) of an infinite $2$-edge-coloured graph $H$. We then notice that several known complexity results for SPs also follow from general complexity classifications of infinite-domain CSPs, suggesting a
  fruitful application of the theory of CSPs to complexity classifications of SPs. We strengthen
  this evidence by using basic tools from constraint satisfaction theory to propose new complexity results 
  of the SP for several graph classes including line graphs of multigraphs, line graphs of bipartite multigraphs,
  $K_k$-free perfect graphs,
  and classes described by forbidding finitely many induced subgraphs, such as 
  $\{I_4,P_4\}$-free graphs, settling an open problem of Alvarado, Dantas, and Rautenbach (2019).
  We also construct a graph sandwich problem which is in coNP, but neither in P nor coNP-complete (unless P = coNP). 

\end{abstract}

\tableofcontents

\section{Introduction}
\label{sec:introduction}

\noindent\textbf{Graph sandwich problems} are computational problems in graph theory
introduced by Golumbic, Kaplan, and Shamir in~\cite{golumbicJA19}.  For a graph class $\calC$, 
the \emph{sandwich problem for ${\mathcal C}$} is the following computational problem, denoted by $\SP({\mathcal C})$:  
The input is a pair of graphs $((V,E_1),(V,E_2))$ where $E_1\subseteq E_2$, and the task
is to decide whether there is an edge set $E$ such that $E_1\subseteq E\subseteq E_2$ and the graph $(V,E)$ belongs to $\calC$. In particular, the SP for $\calC$ is at least as hard as the recognition problem
for $\calC$: a graph $(V,E)$ belongs to $\calC$ if and only if the input $((V,E), (V,E))$ is a yes-instance
of the SP problem for $\calC$. Classifying the computational complexity of the SP for well-structured
hereditary classes has been the main focus of  several papers~\cite{alvaradoAOR280,cameronDM348,dantasDAM182,dantasDAM143,
dantasENTCS346,dantasICCGI07,figueiredoDAM121,golumbicJA19}. 
The sandwich problem for the class of all perfect graphs is not known to be in P and not known
to be NP-hard, and its computational complexity
``remains one of the most prominent open questions in this area''~\cite{cameronDM348}.

The sandwich problem for $\calC$ can be equivalently modeled as follows. The input is a triple
$(V,E,N)$ where $V$ is a set of vertices, $E$ is a set of edges, and $N$ is a set of non-edges. 
The task is to determine whether there is a graph $(V,E') \in \calC$ such that $E\subseteq E'$ and 
$E'\cap N = \varnothing$. We can therefore think of graph sandwich problems as recognition
problems of $2$-edge-coloured graphs, where we think of $E$ as blue edges,
and of $N$ as red edges. We denote by $\SP(\calC)$  the class of $2$-edge-coloured graphs $(V,E,N)$
which are yes-instances to this computational problem for ${\cal C}$; we will also use $\SP(\calC)$ to denote
the sandwich problem for $\calC$ as a computational problem as introduced earlier.

In our context, a \emph{$2$-edge-coloured graph} is a triple $(V,B,R)$ where $(V,R)$ and $(V,B)$
are graphs, and $R\cap B= \varnothing$. A \emph{homomorphism} $f\colon G\to H$ between a pair of
edge-coloured graphs $G$ and $H$ is  a vertex mapping $f\colon V(G)\to V(H)$
that preserves adjacencies and colours, i.e., if $uv\in R(G)$ (resp.\ $uv\in B(G)$), then
$f(u)f(v)\in R(H)$ (resp.\ $B(H)$). In this case,  we write $G\to H$, and if no such homomorphism
exists we write $G\not\to H$. Using notation from constraint satisfaction theory, 
we denote by $\CSP(H)$ the class of finite $2$-edge-coloured graphs $G$ such that $G\to H$.
The \emph{constraint satisfaction problem} of $H$ consisting of deciding whether
an input graph $G$ belongs to $\CSP(H)$. Homomorphisms, and CSPs of uncoloured graphs are
defined analogously. For instance, $\CSP(K_3)$ corresponds to $3$-{\sc Colouring}.
As we will see below, several graph SPs correspond to the CSP of an (infinite) $2$-edge-coloured
graph $H$. Such examples include SP of cographs (i.e., $P_4$-free graphs),
of comparability graphs, of line graphs of bipartite multigraphs, and  of perfect graphs.

\medskip 

\noindent\textbf{Constraint satisfaction theory} can be traced back to Schaefer's classification
of boolean-domain CSPs~\cite{Schaefer}, and to the Hell-Ne\v{s}et\v{r}il theorem classifying the complexity
of CSPs of undirected finite graphs~\cite{HellNesetril} (also known as \emph{graph homomorphism problems}).
This theory provides a uniform approach for classifying the complexity of constraint satisfaction
problems. One of the elementary tools for this approach are \emph{primitive positive constructions}
(shortly introduced), which correspond to gadget reductions (see, e.g.,~\cite{wonderland}), i.e., 
if and only if 
there is a gadget reduction from $\CSP(H')$ to $\CSP(H)$.
This elementary concept turns out to play an important role in the finite-domain dichotomy theorem: if $H$ is a finite
$2$-edge-coloured graph (digraph, or  more generally a relational structure) primitively
positively constructs $K_3$, then $\CSP(H)$ is $\NP$-complete; otherwise, $\CSP(H)$ can 
be solved in polynomial time~\cite{Zhuk20} (originally announced by Bulatov~\cite{BulatovFVConjecture}
and Zhuk~\cite{ZhukFVConjecture}).

A prototypical example of an infinite-domain CSP is the class of acyclic digraphs,
which corresponds to the CSP of the digraph $(\mathbb Q,<)$, i.e., there is an edge
$x\to y$ if and only if $x < y$. Another classical example is the CSP of
the ternary structure $(\mathbb Q, \Betw)$ where $\Betw$ is the \emph{betweenness} relation, 
i.e., the set of triples $(x,y,z)$ with $x < y <z$ or $z < y < x$. The first example
is clearly a tractable example, while the second is a well-known $\NP$-complete infinite-domain CSP.
Notice that both structures enjoy the property that if $f\colon \mathbb{Q\to Q}$ is an 
order preserving bijection, then $f$ is an automorphism of the corresponding structure.
A structure whose vertex set is the set of rational numbers $\mathbb Q$ and satisfies the
previous condition is called a \emph{temporal} structure. It turns out that primitive
positive constructions also play an important role in the classification of temporal CSPs:
if a temporal structure primitively positively constructs $K_3$, then 
$\CSP(H)$ is $\NP$-complete; and otherwise $\CSP(H)$ can be solved in polynomial time~\cite{tcsps-journal}.
Most of the known P versus NP-complete classifications of infinite-domain CSPs
can be also stated in terms of primitive positive constructions~\cite{BMPP16,mottetACM71,posetCSP18}
(for recent examples of classifications of infinite graph and digraph CSPs see, e.g.,~\cite{bodirskySIDMA39,brunarARXIV}).\\

\noindent
\textbf{In this paper}, we establish a connection between graph sandwich problems and
infinite-domain CSPs, which allows us to propose a unified approach to the complexity
classification of several graph sandwich problems. In particular, we will see that
several hardness results about graph SPs can be obtained by means of primitive positive
constructions, resulting in a considerable shortening of previous proofs classifying their
complexity. We will also see that some of the tractable cases of graph SPs can be solved
by algorithms which align with general algorithms for infinite-domain CSPs.

To further illustrate the fruitfulness of the uniform CSP approach to graph SPs, 
we prove new complexity results for graph SPs.
Alvarado, Dantas, and Rautenbach~\cite{dantasAOR188}
studied graph sandwich problems for classes described by two forbidden induced
subgraphs on four vertices. In their classification project 
only eight cases were left open, among them the case ${\mathcal F} = \{I_4,P_4\}$
(where $I_n = nK_1$ denotes the edgeless graph on $n$ vertices).
We prove that the ${\mathcal F}$-free graph sandwich problem is NP-hard. In fact,
this result will be a special case of more general NP-hardness results for infinite classes
of graph sandwich problems, which also yields a complexity classification of
the SP for $K_k$-free perfect graphs (Corollary~\ref{cor:perfect-Kk-free}).
We also present some hardness results for classes of sandwich problems conditional on recent
conjectures from the field of \emph{promise constraint satisfaction problems} and the well-known
Gy\'arf\'as--Sumner conjecture
from graph theory.  Further new complexity results concern the graph sandwich problem for the class
of $\{P_n,K_k\}$-free graphs, of line graphs of multigraphs, and the class of line graphs
of bipartite multigraphs, the last class being closely related to perfect graphs whose SP
complexity is the main open problem in the area.

Inspired by results from constraint satisfaction, we also identify a graph sandwich problem
which is coNP-complete (Section~\ref{sect:non-dicho}). By adapting the proof of Ladner's theorem
we even find graph sandwich problems that are \emph{coNP-intermediate}, i.e., problems in
coNP and neither in P nor coNP-complete (unless P = coNP).
For us, this result complements the complexity classifications mentioned above by noticing that
graph sandwich problems  capture complexity classes beyond our current understanding, and thus, any
systematic approach to complexity classifications must impose certain restrictions on which
GSPs to consider. It remains open whether there are NP-intermediate graph sandwich problems.

The rest of this paper is structured as follows. In Section~\ref{sec:preliminaries},
we provide standard notions from graph and model theory, and introduce the
basic tools from constraint satisfaction theory needed for this work.
In Section~\ref{sec:CSP}, we establish some simple conditions upon a graph class $\calC$ so that the
SP for $\calC$ corresponds to the homomorphism problem of a $2$-edge-coloured graph $H$.
To illustrate the usefulness of the connection between graph sandwich problems and CSPs further,
we consider several hereditary properties for which the complexity of their SP has been
classified before, and see that these complexity classifications can be easily settled
via results from CSPs (Section~\ref{sect:known}).
In Sections~\ref{sect:line}, \ref{sect:salt-free}, and~\ref{sect:Gs--SP} we present new
complexity for the graph sandwich problems, including previously open cases from the literature.
Our non-dichotomy result can be found in Section~\ref{sect:non-dicho}. 
We conclude in Section~\ref{sec:conclusions} with a list of open problems towards a more systematic
understanding of the complexity of graph sandwich problems.

\section{Preliminaries}
\label{sec:preliminaries}

All graphs in this paper are simple graphs, unless we specifically talk about \emph{multigraphs},
i.e., graphs with possibly multiple edges between the same pair of vertices. All edge-coloured graphs
will also be simple, i.e., there are no multiple monochromatic edges
between two vertices $x,y$ (notice that this allows for multiple heterochromatic edges). 
We follow standard notions form graph theory (see, e.g.,~\cite{bondy2008}). In particular,
we denote by $\overline G$ the \emph{complement} of a graph $G$, i.e., $V(\overline G) = V(G)$,
and $E(\overline G) = \{xy\colon x\neq y, xy\not\in E(G)\}$. We will denote by $N(G)$ the set
of \emph{non-edges} of $G$, i.e., $N(G) = E(\overline G)$,
as usual, when $G$ is clear from context, 
we simply write $N$ instead of $N(G)$.

We also assume familiarity with basic first-order logic, and refer the
reader, e.g., to~\cite{HodgesLong}.

\subsection{Graph Classes and Universal Graphs}
All graph classes $\calC$ considered in this paper are closed under \emph{isomorphism}, i.e., 
if $H$ belongs to $\calC$ and there is a bijection $f\colon V(G)\to V(H)$
preserving the edge relation, and the non-edge relation, then $G$ belongs to $\calC$ as well. 
We say that $\calC$ has the \emph{hereditary property} (or simply, that $\calC$ is a
\emph{hereditary} class) if it is closed under taking induced subgraphs, 
i.e., if $H$ belongs to $\calC$, and $G$ is obtained from $H$ by removing some vertices, then
$G$ belongs to $\calC$ as well.  Finally, we say that $\calC$ has the \emph{joint embedding property (JEP)}
if for every pair of graphs $G$ and $H$ in $\calC$, there is a graph $F$ in $\calC$ such
that $G$ and $H$ are induced subgraphs of $F$. 

A graph $G$ \emph{embeds} into a graph $H$ if $G$ is isomorphic to an induced subgraph
of $H$. Given a (possibly infinite) graph $H$, we denote by $\Age(H)$ the class
of finite graphs that embed into $H$, and we call it the \emph{age} of $H$.
Notice that if a class $\calC$ corresponds to the age of some graph $H$, then 
$\calC$ is hereditary and has the joint embedding property. Actually, it is a
well-known model theoretic fact that a class $\calC$ is the age of some countable
graph $H$ if and only if $\calC$ is hereditary and has JEP~\cite[Theorem 7.1.1]{HodgesLong}
(see also~\cite[Theorem 2]{scheinermanDM55} for a graph theoretic version in terms of 
injective intersection classes).

For the purposes of this paper, we say that a graph $H$ is \emph{universal} in
a graph class $\calC$ if $H$ belongs to $\mathcal C$, and $\Age(H)$ coincides with
the class of finite graphs in $\calC$. 
If the class $\mathcal C$ is called ``name'', we will say that
$H$ is a universal ``name'' graph, e.g., we will talk about a universal bipartite graph,
a universal split graph, and a universal threshold graph.

\begin{example}\label{ex:random-graph}
   
        The \emph{Random Graph} is the graph with vertex set $\mathbb N$
        and for every pair of vertices $u,v\in \mathbb N$ there is an
        edge between $u$ and $v$ independently at random with probability $1/2$. This graph is
       a universal graph and unique up to isomorphism 
       (and is also called the \emph{Rado Graph};  
        it can also be constructed deterministically~\cite[Theorem 7.4.4]{HodgesLong}).
    \end{example}

        \begin{example}\label{ex:random-bipartite-graph}
        The \emph{Random Bipartite Graph}, denoted $\Bip$, is the graph whose vertex set
        consists of two disjoint countably infinite sets $X$ and $Y$ such that between each pair of vertices
        $x\in X$ and $y\in Y$ there is an edge $xy\in E$  independently at random with probability
        $1/2$; there are no edges with $X$ and within $Y$, so $x$ and $Y$ are independent sets in the graph theoretic sense.
        Again, this graph is unique up to isomorphism, and is a universal bipartite graph.
        It can also be constructed deterministically by standard model-theoretic constructions (see, e.g.,~\cite{Hodges}).
\end{example}

\subsection{Constraint Satisfaction Problems}
A \emph{relational signature} is a set $\tau$ of relation symbols $R,S,\dots$
each equipped with a positive integer $r$ called its \emph{arity}. A \emph{$\tau$-structure}
$H$ consists of a \emph{vertex set} $V(H)$ (also called the \emph{domain}), 
and for each $R\in \tau$ an $r$-ary relation $R(H)\subseteq V(H)^r$ called
the \emph{interpretation} of $R$ in $H$. Following this formalism,
a digraph is an $\{E\}$-structure where $E$ is a binary relation symbol, 
a graph is an $\{E\}$-structure $G$ where the interpretation of $E$ in 
$G$ is a symmetric (binary) relation, and a $2$-edge-coloured graph
is an $\{R,B\}$-structure $H$ where each interpretation $R(H)$ and $B(H)$
is a symmetric binary relation.

A \emph{homomorphism} between a pair of $\tau$-structures $A,B$
is a vertex mapping $f\colon V(A)\to V(B)$ such that for each $R\in \tau$  of
arity $r$, and every $r$-tuple $(v_1,\dots, v_r)\in R(A)$ the image
$(f(v_1),\dots, f(v_r))$ belongs to the interpretation $R(B)$. Naturally, 
this notion of homomorphism coincides with the previously introduced
notion of homomorphism between edge-coloured graphs, i.e., vertex
mappings $f\colon V(G)\to V(H)$ preserving adjacencies, and respecting
the colour of the edges.

The \emph{product} of a pair of $\tau$-structure $A$ and $B$ is the
$\tau$-structure $A\times B$ with vertex set $V(A)\times V(B)$, and
for each $R\in \tau$ of arity $r$ there is tuple $((a_1,b_1),\dots,
(a_r,b_r))$ in the interpretation of $R$ in $A\times B$ if and only if 
$(a_1,\dots, a_r)\in R(A)$ and $(b_1,\dots, b_r)\in R(B)$. For a positive
integer $k$ we denote by $A^k$ the product $A\times A^{k-1}$ where
$A^1 = A$.

The \emph{disjoint union} of a pair of $\tau$-structures $A$ and $B$
with $V(A)\cap V(B) = \varnothing$, is the structure $A+ B$
with domain $V(A)\cup V(B)$, and for each $R\in\tau$ the interpretation
$R(A+ B)$ is the union $R(A)\cup R(B)$. Notice that
if $A,B\in \CSP(C)$, then $A + B\in \CSP(C)$. It is straightforward
to observe that a class $\calC$ of finite structures is the CSP of some
(possibly infinite) countable structure $A$ if and only if $\calC$ is
closed under disjoin unions, and under inverse homomorphisms, i.e.,
if $B\in \calC$ and $A\to B$,  then $A\in \calC$ (see, e.g.,~\cite{Book}).

The \emph{constraint satisfaction problem} (CSP) of a $\tau$-structure
$A$ consists of deciding whether an input (finite) $\tau$-structure
$B$ maps homomorphically  to $A$. We denote this problem by $\CSP(A)$,
and use the same notation $\CSP(A)$ to denote the class of finite
structure $B$ that map homomorphically to $B$, i.e., the class
of yes-instances to the CSP of $A$.

\subsection{Primitive Positive Constructions}

A standard reduction between CSPs known to many graph theorists is the
following gadget reduction from $\CSP(K_5)$ (i.e., $5$-{\sc Colouring})
to $\CSP(C_5)$. Given an input $G$ to $\CSP(K_5)$ construct
a new graph $G'$ by replacing each edge $xy\in E(G)$ by a path of
length three $xuvy$ where the vertices $u$ and $v$ are fresh new
vertices for each  edge $xy\in E(G)$. It is standard procedure to verify that
$G$ admits a $5$-colouring if and only if the new graph $G'$
admits a homomorphism to $C_5$. Now, consider the binary relation
$E'$ defined on the vertices of $C_5$ by ``there is a path of length
$3$ between $x$ and $y$''. The reader can easily verify that the
graph $(V(C_5), E')$ is the complete graph on $5$ vertices. Primitive positive
constructions are a generalization of the previous definition of $K_5$ in $C_5$,
and in a similar way that the previous definition is  connected
to the gadget reduction above, general primitive positive
constructions also yield gadget reductions for general CSPs.

A \emph{primitive positive (pp) formula} $\varphi(\bar x)$ of graphs is a 
formula with existentially quantified variables $\bar y$ and 
whose quantifier-free part is a conjunction of equalities
$u = v$ and positive atoms $E(w,z)$ where each variable  $u,v,z,w$
is one of the free variables from $\bar x$ or one the existentially quantified
variables from $\bar y$.  For instance, 
the relation $E'(x,y)$ from the example above admits the following
pp-definition
\[
    E'(x,y):=\exists z,w (E(x,z)\land E(z,w)\land E(w,y)). 
\]
Primitive positive $\tau$-formulas are defined analogously, i.e., 
existentially quantified formulas whose quantifier-free part
is a conjunction of positive atoms from $\tau$ and equalities. 

A \emph{primitive positive definition} of digraphs in a signature
$\tau$ of dimension $d\in \mathbb Z^+$ is a primitive positive formula
$\delta_E(x_1,\dots, x_d, y_1,\dots, y_d)$ with $2d$ many free variables.
For every such primitive positive definition we associate a mapping  $\Pi$
from $\tau$-structures to digraphs as follows.
Given a $\tau$-structure $A$ the \emph{pp-power} $\Pi(A)$ of $A$
is the digraph
\begin{itemize}
    \item  with vertex set $V(A)^d$, and
    \item there is an edge $(a_1,\dots, a_d) \to (b_1,\dots, b_d)$ if and only if
    $A\models \delta_E(a_1,\dots, a_d,b_1,\dots, b_d)$.
\end{itemize}
We say that a structure $A$ \emph{pp-constructs} a digraph $D$ if there is a
primitive positive definition $\delta_E$ such that $\Pi(A) \to D \to  \Pi(A)$, 
i.e., $\Pi(A)$ is homomorphically equivalent to the digraph $D$. Continuing with our previous
example, consider the $1$-dimensional primitive positive definition
\[
\delta_E(x,y):=\exists z,w \; \big (E(x,z)\land E(z,w)\land E(w,y) \big),
\] 
then the pp-power $\Pi(C_5)$ of the $5$-cycle is the complete graph
$K_5$ (see also Figure~\ref{fig:C5-K5}).

\begin{figure}
\centering
\begin{tikzpicture}

  \begin{scope}[xshift = -6cm, scale = 0.6]
    \node [vertex] (1) at (90:2) {};
    \node [vertex] (2) at (18:2) {};
    \node [vertex] (3) at (306:2) {};
    \node [vertex] (4) at (234:2) {};
    \node [vertex] (5) at (162:2) {};
    \node (L1) at (0,-2.75) {$C_5$};
      
    \foreach \from/\to in {1/2, 2/3, 3/4, 4/5, 5/1}     
    \draw [edge] (\from) to  (\to);
  \end{scope}

  \begin{scope}[scale = 0.6]
    \node [vertex] (1) at (90:2) {};
    \node [vertex] (2) at (18:2) {};
    \node [vertex] (3) at (306:2) {};
    \node [vertex] (4) at (234:2) {};
    \node [vertex] (5) at (162:2) {};
    \node (L1) at (0,-2.75) {$\Pi(C_5)\cong K_5$};
      
    \foreach \from/\to in {1/2, 2/3, 3/4, 4/5, 5/1, 1/3, 3/5, 5/2, 2/4, 4/1}     
    \draw [edge] (\from) to  (\to);
  \end{scope}

\end{tikzpicture}
\caption{Consider  the primitive positive definition (of $\{E\}$ in $\{E\}$) 
where $\delta_E(x,y):=\exists z,w\; \big ( E(x,z)\land E(z,w)\land E(w,y) \big )$. 
Here, we depict $C_5$ and its pp-power $\Pi(C_5)\cong K_5$.
}
\label{fig:C5-K5}
\end{figure}
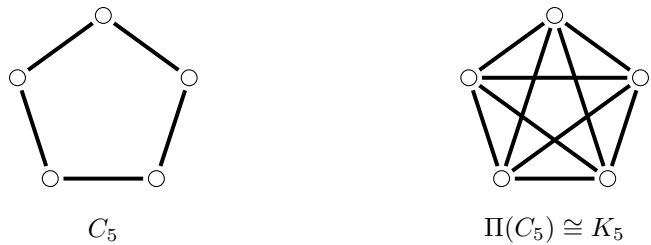

In general, a primitive positive definition of a signature
$\sigma$ in a signature $\tau$ is a set $\Delta$ of primitive positive
formulas $\delta_R$ indexed by relation symbols $R\in \sigma$, and such
that each formula $\delta_R$ has exactly $r\cdot d$ free variables, where
$r$ is the arity of $R\in \tau$. Pp-powers and pp-constructions of
general relational structures are then defined accordingly.

As mentioned above, pp-constructions yield a canonical class of log-space reductions
between constraint satisfaction problems. 

\begin{lemma}[Corollary 3.5 in \cite{wonderland}]
\label{lem:pp-construction-reduction}
    Let $A$ and $B$ be (possibly infinite) structures with finite relational
    signature. If $A$ pp-constructs $B$, then $\CSP(B)$ reduces in logarithmic space
    to $\CSP(A)$. 
\end{lemma}

\section{Sandwich Problems and CSPs}
\label{sec:CSP}
Let us begin with a simple example. A graph $G$ is called a \emph{split graph}
if its vertex set can be covered by an independent set $I$ and
a set $K$ inducing a complete subgraph. Consider now the countable
graph $S = (X\cup Y,E)$ obtained from the Random Bipartite graph $\Bip = (X,Y,E(\Bip))$
(see Example~\ref{ex:random-bipartite-graph}) by adding to $E(\Bip)$ all edges between
pairs of distinct vertices in $Y$, i.e., $E:=E(\Bip)\cup \{yy'\colon y,y'\in Y, y\neq y'\}$.
Using the fact that the Random Bipartite graph is a universal graph  for the class of all
bipartite graphs,  one can notice that $S$ is a universal graph for the class of all split graphs.
Now, consider the $2$-edge-coloured graph $(X\cup Y,B,R)$ where
 the set of blue edges $B$ is $E$, and the set of red edges
is the set of non-edges of $S$, i.e., the set of edges of $\overline S$.
It follows that $\CSP(X\cup Y,B,R)$ is the sandwich problem for the class of split graphs.
Indeed, 
if an input $(V,E,N)$ to the SP of split graphs is a yes-instance, 
then there is an edge set $E'$ such that $(V,E')$ is a split graph, and
so there is an embedding $f\colon (V,E')\to S$, and this embedding
also defines a homomorphism of $2$-edge-coloured graphs
$f\colon (V,E,N)\to (X \cup Y,B,R)$. Conversely, suppose that there
is a homomorphism $f'\colon (V,E,N)\to (X \cup Y,B,R)$, and consider the edge set
\[
    E' = \{uv\colon f'(u),f'(v)\in Y\}~\cup~\{uv\colon f'(u)f'(v)\in B\}.
\]
It is not hard to observe that $(V,E')$ is a split graph. Also,
since $f'$ maps adjacent vertices $xy\in E$ to blue edges of $(X \cup Y,B,R)$, 
it follows that $E\subseteq E'$. One can similarly notice that
if $uv\in N$, then $f'(u)f'(v)\in R$, and in particular at least
one of the vertices $f'(u)$ or $f'(v)$ belongs to $X$, hence
$E'\cap N = \varnothing$. Therefore, $(V,E,N)$ is a yes-instance
to the SP for split graphs if and only if $(V,E,N)$ is a yes-instance to $\CSP(X \cup Y,B,R)$.

The aim of this section is to generalize the previous example, and 
to characterize those graph classes $\calC$ whose sandwich problem
is the CSP of some $2$-edge-coloured graph. 
(As a corollary, we also obtain a characterization of those CSPs of
$2$-edge-coloured graphs that are graph sandwich problems (Corollary~\ref{cor:CSP-SP}).
To do so, it will be convenient
to first consider the injective variant of CSPs. 
Given a $2$-edge-coloured graph $(V,B,R)$, 
the \emph{injective CSP of $(V,B,R)$} is 
the class of $2$-edge-coloured graphs that admit
an injective homomorphism into $(V,B,R)$, and will be denoted by $\injCSP(V,B,R)$. 

\begin{definition}
    We say that a $2$-edge-coloured graph $H$ is \emph{complete}
    if for every pair of vertices $u,v\in V$ there is exactly one edge
    between $u$ and $v$, i.e., $uv\in B$ and $uv\not\in R$, or vice versa.
    For a graph $H$ we denote by $H^\ast$ the complete $2$-edge-coloured graph
    with $V(H) = V(H^\ast)$, $B(H^\ast) = E(H)$, and $R(H^\ast) = N(H)$. 
\end{definition}

   Note that
    every complete $2$-edge-coloured graph is of the form $H^\ast$
    for some graph $H$.

\begin{remark}\label{rmk:injCSP}
For every  graph $H$ the injective CSP of $H^\ast$ is the sandwich problem for $\Age(H)$.
\end{remark}

Two vertices $u,v$ of a graph $G$ are called \emph{twins} if $N(u) = N(v)$,
and they are called \emph{co-twins} if $uv\in E$ and
$N(u)\setminus\{v\} = N(v)\setminus\{v\}$.We deviate from some the terminology in some references
where twins and co-twins are called false and true twins, respectively. 
Notice that $u$ and $v$ are co-twins in $G$ if and only if they are twins in
$\overline G$. We say that $G'$ is a \emph{blow-up} 
of $G$ if $G'$ can be obtained from $G$ by iteratively adding to $G$ a twin of some vertex of $G$. 
Analogously, we define the concept of a \emph{co-blow-up} of $G$.  
For instance, $C_4$ is
a blow-up of $P_3$, and $K_5$ is a co-blowup of $K_3$.

We say that a class of graphs $\calC$ is \emph{preserved} by \emph{blow-ups} (resp.\ by
\emph{co-blow-ups})
if for every graph $G$ in $\calC$, each blow-up $G'$ (resp.\ each co-blow-up)
of $G$ belongs to $\calC$ as well.
The reader familiar with graph theory may notice that several graph
classes are preserved by blow-ups or co-blow-ups; some examples include
cographs (preserved by blow-ups and co-blow-ups), perfect graphs
(preserved by blow-ups and co-blow-ups), chordal graphs (preserved by
co-blow-ups, but not by blow-ups), interval graphs (preserved by co-blow-ups,
but not by blow-ups), and comparability graphs (preserved by blow-ups and co-blow-ups).

Consider a graph $G$ and a partition $(I,C)$ of its vertex set (with possibly one empty part).
A \emph{split blow-up} of $G$ \emph{with respect to} $(I,C)$ is a graph $G'$ obtained 
from $G$ be iteratively substituting vertices $x\in I$ by a set of twins, and vertices 
$y\in C$ by a set of co-twins --- $I$ stands for ``independent'' meaning that a vertex
$x\in I$ is blown-up to an independent set, similarly, $C$ stands for ``complete''. 
We say that a class $\calC$ is preserved by \emph{split blow-ups} if for every
graph $G$ in $\calC$, there is a partition $(I,C)$ (with possibly one empty part) 
of the vertices of $G$ such that any split blow-up of $G$ with respect to $(I,C)$
belongs to $\calC$ as well. 
In particular, 
if $\calC$ is preserved by blow-ups (or is preserved by co-blow-ups), then $\calC$
is preserved by split blow-ups. 
The name of split blow-ups is clearly motivated by split graphs: the class of 
split graphs is a natural example of a class preserved under split blow-ups. Also
notice that it is neither preserved under blow-ups nor under co-blow-ups:
$2K_2$ is not a split graph, and it is a co-blow-up of $2K_1$ which is a split graph, 
and $C_4$ is not a split graph and it is a blow-up of $P_3$ which is a split graph.

\begin{lemma}\label{lem:H->H*}
    Let $\calC$ be a hereditary class of graphs with the joint embedding property and
    which is preserved under split blow-ups,   
    and let $H$ be a universal graph in $\calC$. Then $\SP(\calC) = 
    \injCSP(H^{\ast}) = \CSP(H^\ast)$.
\end{lemma}
\begin{proof}
    The equality $\SP({\mathcal C}) = \injCSP(H^*)$ follows from Remark~\ref{rmk:injCSP}.
    Let $(V,E,N)$ be a $2$-edge-coloured graph. 
    We show that
    $(V,E,N)$ is a yes-instance of the SP for $\calC$ if and only
    if $(V,E,N)\to H^\ast$. 
    If $(V,E,N)$ is a yes-instance to the SP for $\calC$, then
    $(V,E,N) \in \injCSP(H^\ast) \subseteq \CSP(H^\ast)$. 
    Conversely, suppose that there is a homomorphism $f\colon (V,E,N)\to H^\ast$. 
    Let $U$ be the image of $V$. By the definition of $H$,
    it follows that the (finite) subgraph $G$ of $H$ induced by
    $U$ belongs to $\calC$. We now construct an edge set $E'$ over
    the vertex set $V$ as follows. If $f(u)\neq f(v)$, then $uv\in E'$
    if and only if $f(u)f(v)\in E(G)$; if $f(u) = f(v)$, then
    $uv\not\in E'$ if the blow-up obtained from $G$ by substituting $f(u)$
    by a set of twins belongs to $\calC$, and $uv\in E$ if the co-blow-up
    obtained from $G$ by substituting $f(u)$ by a set of co-twins belongs to 
    $\calC$. The edge set $E'$ is well-defined since $\calC$ is preserved under
    split blow ups, and using the fact that $f$ is a homomorphism, we conclude
    that $E\subseteq E'$ and $N\cap E' =\varnothing$. Therefore,
    $\CSP(H^\ast) = \SP(\calC)$.
\end{proof}

We are now ready to characterize those graph classes whose sandwich problem
is the CSP of some $2$-edge-coloured graph.

\begin{proposition}\label{prop:SP->CSP}
    The following statements are equivalent for a class of 
    graphs $\calC$.
    \begin{enumerate}
        \item $\calC$ is a hereditary class, with the joint embedding property, 
        and preserved under split blow-ups.
        \item There is a $2$-edge-coloured graph $(V,R,B)$ such that
        $\CSP(V,R,B) = \SP(\calC)$.
        \item There is a graph $H$ such that 
         $\SP(\calC) = \CSP(H^\ast) = \injCSP(H^\ast)$.
    \end{enumerate}
\end{proposition}
\begin{proof}
    If the first statement holds, and $H$ is a universal graph 
    in the class $\calC$, then we have $\CSP(H^\ast) = 
    \injCSP(H^{\ast}) = \SP(\calC)$ by Lemma~\ref{lem:H->H*}. So the first item
    implies the third one.
    The third item clearly implies the second one.
    We now show that the second item implies the first one. 
    
    We first show that $\calC$ is a hereditary class. Suppose that
    $G\in \calC$, and let $H$ be an induced subgraph of $G$.
    Clearly, $H^\ast\to G^\ast$, so if $G\in \calC$, then $G^\ast$
    is a yes-instance to the SP for $\calC$, which by the second item
    implies that $G^\ast\to (V,R,B)$.
    Thus, $H^\ast\to (V,R,B)$, which shows that $H\in \calC$. 
    
    Now suppose that $G,H \in {\mathcal C}$ are 
    graphs with disjoint vertex sets. We then have $G^\ast,H^\ast \in \CSP(V,R,B)$,
    and $G^\ast + H^\ast\in \CSP(V,R,B)$, because CSPs are closed under disjoint unions. 
    Notice that $G^\ast+ H^\ast$ is an instance to 
    $\SP(V(G)\cup V(H), E(G)\cup E(H), N(G)\cup N(H))$, and thus
    there is an edge set $E'$ such that $(V(G)\cup V(H), E')\in \calC$
    where $E(G)\cup E(H) \subseteq E'$, and $E\cap (N(G)\cap N(H)) = \varnothing$. 
    Therefore, $G$ and $H$ are induced subgraphs of $(V,E')\in \calC$,
    proving that $\calC$ has the JEP. 
    
    We finally
    show that $\calC$ is preserved under split blow ups. To do so, it
    suffices to show that if $G\in \calC$, $v\in V(G)$, and $n$ is a
    positive integer, then the graph $G'$ obtained from $G$ be substituting
    $v$ by $n$ twins or by $n$ co-twins also belongs to $\calC$.
    Let $r$ be larger than the Ramsey number of $n$, i.e., any $2$-edge-coloured
    graph on $r$ vertices contains either a blue or a red clique on $n$ vertices. 
    Again, we know that $G^\ast\to (V,R,B)$. 
    Let $G^\ast_r$ be
    the $2$-edge coloured graph obtained from $G^\ast$ be blowing up
    $v$ to an independent set $I$ of size $r$ such that $u\in I$
    and $w\in V(G^\ast)\setminus I$ there is a red edge $uw$ if and only
    if there is a red edge $vw$, and there is a blue edge $uw$ if and
     only if there is a blue edge $vw$. Clearly, $G^\ast_r\to G^\ast\to (V,R,B)$, 
     and thus $G^\ast_r$ is a yes-instance to the SP for $\calC$. Let
     $E'$ be an edge set such that $(V(G^\ast_r), E')\in \calC$, 
     $B(G^\ast_r)\subseteq E'$, and $E'\cap R(G^\ast_r) = \varnothing$.
     By the definition of $r$, the subgraph
     of $(V(G^\ast_r),E')$ induced by $I$ contains a subset $I'$
     on $n$ vertices that induces either a clique or a complete subgraph.
     Since the only undetermined edges in $G^\ast_r$ are between
     pairs of vertices in $I$, we see that for $w\in V(G)\setminus I'$
     there is an edge $uw\in E'$ if and only if $vw\in E'$. We thus 
     conclude that the subgraph $G'$ of $(G^\ast_r,E')$ induced by
     $V(G)\cup I'$ is the desired blow-up or co-blow-up of $G$
     (here we also used the already proved fact that $\calC$ is a hereditary
     class). This concludes the proof.
\end{proof}

Building on this proposition we can also characterize those constraint satisfaction problems
that are graph sandwich problems.

\begin{corollary}\label{cor:CSP-SP}
    The $\CSP$ of a $2$-edge-coloured graph $(V,B,R)$ equals the SP of some graph class $\calC$
    if and only if there is a graph $H$ such that
        $$\CSP(V,B,R) = \CSP(H^\ast) = \injCSP(H^\ast).$$
\end{corollary}
\begin{proof}
    If $\CSP(V,B,R) = \injCSP(H^*)$, then $\CSP(V,B,R) = \SP(\Age(H))$ 
    by Remark~\ref{rmk:injCSP}. Conversely, if $\CSP(V,B,R)$ equals $\SP({\mathcal C})$ for
    some graph class $\mathcal C$, then, by the equivalence between the second and 
    third statements in Proposition~\ref{prop:SP->CSP}, there is a graph $H$ such that 
    $\SP(\calC) = \CSP(H^\ast) = \injCSP(H^\ast)$. Since $\CSP(V,B,R) = \SP(\calC)$, 
    this implies the statement. 
\end{proof}

A \emph{minimal obstruction} to a hereditary class $\calC$ is a graph
$G\not\in \calC$ such that each proper induced subgraph of $G$ belongs
to $\calC$. It is well-known and easy to observe that every hereditary
class of graphs is uniquely determined by its set of minimal obstructions.
A graph $G$ is called \emph{point-determining} if $G$ contains no pair
of twins (e.g., odd-holes and odd-anti-holes are point-determining graphs).
It is not hard to see that if $\calF$ is a set of point-determining
graphs, then the class of $\calF$-free graphs is preserved by blow-ups
(e.g., the class of perfect graphs is preserved by blow-ups). 

\begin{corollary}\label{cor:examples}
    Let $\calC$ be a hereditary class of graphs with the joint embedding property. 
    If the minimal obstructions of $\calC$ are point-determining graphs, then 
    the sandwich problem for $\calC$ is a CSP. In particular, the SP for the class of 
    perfect graphs is a CSP. 
\end{corollary}

\section{Known Examples}
\label{sect:known} 

In this section we provide a (probably incomplete) list of graph sandwich
problems whose computational complexity has been classified in the literature,
and which are CSPs. We also revisit some of these complexity classifications
through the lens of constraint satisfaction theory with the intent of building
intuition and providing examples of the CSP approach to sandwich problems. However, 
this section can be safely skipped by many readers.

\begin{example}\label{ex:literature}
    The following is a list of graph classes $\calC$ such that
    the complexity of $\SP(\calC)$ has been classified, 
    and moreover the sandwich problem is a CSP (the latter
    can be verified via Proposition~\ref{prop:SP->CSP}).
    \begin{multicols}{3}
    \begin{itemize}[itemsep = 0.2pt]{\small
        \item Chordal~\cite{golumbicJA19}.
        \item Comparability~\cite{golumbicJA19}.
        \item Circle~\cite{golumbicJA19}.
        \item Path graphs~\cite{golumbicJA19}.
        \item Directed path graphs~\cite{golumbicJA19}.
        \item Split~\cite{golumbicJA19}.
        \item Permutation~\cite{golumbicJA19}.
        \item Trivially perfect~\cite{alvaradoAOR280}.
        \item Cographs~\cite{golumbicJA19}. 
        \item Circular-arc~\cite{golumbicJA19}. 
        \item Interval graphs~\cite{golumbicJA19}.
        \item Proper circular-arc~\cite{golumbicJA19}.
        \item Unit circular-arc~\cite{golumbicJA19}.
        \item Proper interval~\cite{golumbicJA19}.
        \item Threshold~\cite{golumbicJA19}.
        \item $(k,l)$-graphs~\cite{dantasDAM143}.
        \item Clique-helly~\cite{douradoJBCS14}.
        \item Hereditary clique-helly~\cite{douradoJBCS14}.
        \item Strongly chordal~\cite{figueiredoTCS381}
        \item Bipartite chain~\cite{dantasAOR188}.
        \item Odd-hole-free~\cite{cameronDM348}.
        \item Even-hole-free~\cite{cameronDM348}. 
        \item $3PC(\cdot,\cdot)$-free~\cite{dantasDAM159}.
        \item $C_n$-free~\cite{dantasDAM159}.
        \item Complete multipartite~\cite{dantasDAM159}.
        \item $P_n$-free~\cite{figueiredoDAM251}.}
    \end{itemize}
     \end{multicols}
\end{example}

In the rest of this section we revisit 
a few sandwich problems listed in Example~\ref{ex:literature}.
Our objective is to argue that  the seemingly ad-hoc complexity
classifications of certain sandwich problems  actually align with
general uniform approaches to complexity classifications in the world
of CSPs. 

\begin{itemize}
    \item 
     \textbf{Complete multipartite graphs} (Section~\ref{sect:multipartite}). 
     We show that this sandwich problem
    can be solved by a \emph{Datalog program}. Datalog is another tool to solve
    infinite-domain CSPs in polynomial time (see, e.g.,~\cite[Chapter 8]{Book}).
    We also show that the algorithm from~\cite{dantasDAM159} used to solve the
    SP for $(K_n-e)$-free graphs can be formulated as a Datalog Program as well.
    \item 
    \textbf{Split graphs} (Section~\ref{sect:split}).   
    This SP falls into a family of CSPs expressible
    in \emph{Monotone Monadic Strict NP (MMSNP)} whose complexity has been classified~\cite{MMSNP-Journal}.
    Moreover, the algorithm proposed in~\cite{golumbicJA19} to solve the SP for split graphs
    coincides with a prominent technique used to solved CSPs in the previous family
    sometimes called \emph{reduction to the finite}~\cite[Section 10.5.3]{Book}.
    \item 
    \textbf{Threshold graphs} (Section~\ref{sect:threshold}). 
    We use this example
    to illustrate the \emph{algebraic approach} to CSPs. This approach provides
    some abstract algebraic conditions for a CSP to be solved by certain
    polynomial-time algorithms, e.g.,   by \emph{consistency checking}, or 
    algebraic conditions guaranteeing that the corresponding CSP is $\NP$-complete.
    We will see that the polynomial-time tractability of the SP for threshold graphs can be explained by such
    algebraic conditions. 
    \item 
    \textbf{Comparability graphs} (Section~\ref{sect:comparability}).      
    We consider this example
    of a hard sandwich problem, to illustrate an application of a classification
    of certain reducts of \emph{finitely bounded homogeneous} structures; classifying
    the complexity of such CSPs is a very active research line in the theory of infinite-domain
    CSPs, in particular, it is conjectured that these CSPs exhibit a P versus NP-complete
    dichotomy~\cite[Conjecture 3.7.1]{Book}.
    In the case of the SP for the class of comparability graphs, the NP-hardness follows from
    a classification result of Kompatscher and van Pham~\cite{posetCSP18}.
    \item 
    \textbf{Generalized split graphs} (Section~\ref{sect:gensplit}). 
    We consider the P vs.\ NP-complete classification of the sandwich problem
    for $(p,q)$-split graphs from~\cite{dantasDAM143}. We reprove the hard cases
    to illustrate a hardness proof of an infinite-domain CSP via a primitive
    positive construction of a finite structure with a hard CSP.
    \item 
    \textbf{Permutation graphs} (Section~\ref{sect:permutation}). 
    Finally, we consider the sandwich problem for the class of permutation graphs to propose
    a translation of a gadget proof from~\cite{golumbicJA19} to a primitive positive
    construction,  of an infinite-domain structure with a hard CSP, namely the temporal
    constraint satisfaction problem $(\mathbb Q, \Betw)$
    (defined in Section~\ref{sec:introduction})).
\end{itemize}

\subsection{Complete Multipartite Graphs and Datalog}
\label{sect:multipartite}

A graph $G$ is a \emph{complete multipartite graph} if there is a partition of $V$
into independent sets $V_1,\dots, V_k$, such that every pair of vertices in different
parts $V_i$, $V_j$ are adjacent. 
Consider the graph $K_{\omega}[I_{\omega}]$ obtained from the infinite
clique (with vertex set $\mathbb N$) by blowing up each vertex to a countable
set of twins. It is straightforward to observe that a graph $G$
is a complete multipartite graph if and only if it embeds into $K_{\omega}[I_{\omega}]$.
A simple way to solve the CSP of $K_{\omega}[I_{\omega}]^\ast$ is with the following procedure
on input $(V,E,N)$:
\begin{enumerate}
    \item initiate $E' = \varnothing$,
    \item for each $xy\in E$, add $xy$ to $E'$,
    \item for each $xy,yz\in E'$ add $yz$ to $E'$,
    \item if there is $xz\in E'\cap N$, reject,
    \item repeat 3, 4 until $E'$ does not change.
\end{enumerate}
It is not hard to see that if the program rejects, then $(V,E,N)\not\in\CSP(K_{\omega}[I_{\omega}]^\ast)$, 
and if it does not reject, then  $(V,E')$ is
a complete multipartite graph, such that $E\subseteq E'$ and $E'\cap N = \varnothing$, 
i.e., $(V,E,N)\in \CSP(K_{\omega}[I_{\omega}]^\ast)$.

The previous procedure is a particular example of a \emph{Datalog Program}\footnote{The
definition of Datalog is not needed for this work; we refer the interested reader
to~\cite[Chapter 8]{Book}.} which is often used to solve finite- and infinite-domain CSPs
(see, e.g.,~\cite[Chapter 8]{Book}). By strengthening the power of these programs
one can solve other problems that are not necessarily CSPs.
One possible way to enrich the power of Datalog is by allowing
inequalities in the definition of $E'$. 
For instance, the SP for $(K_n-e)$-free graphs is not a CSP whenever $n \ge 4$:
this can be easily seen via Proposition~\ref{prop:SP->CSP} by noticing 
that this class is not closed under split blow-ups.
The polynomial-time algorithm proposed by~\cite{dantasDAM159} to solve the SP for $(K_n-e)$-free graphs 
can be naturally encoded as a Datalog Program enriched with inequalities.

\subsection{Split Graphs and Reduction to the Finite}
\label{sect:split}
We first consider the following toy problem.
The task is to decide whether the
vertices of an input graph $G$ can be coloured with two colours $0,1$
in such a way that you avoid the following colourings of $K_4$, where filled
vertices represent colour $1$ and non-filled ones represent colour $0$.
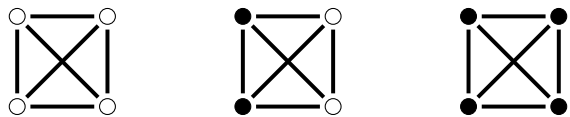
\begin{center}
    \begin{tikzpicture}[scale = 0.6]
    \begin{scope}[xshift=-5cm]
        \node [vertex] (1) at (-1,0) {};
        \node [vertex] (2) at (-1,2) {};
        \node [vertex] (3) at (1,2) {};
        \node [vertex]  (4) at (1,0) {};
        
        \foreach \from/\to in {1/2, 2/3, 3/4, 4/1, 1/3, 2/4} 
        \draw [edge] (\from) to (\to);

    \end{scope}
    \begin{scope}
        \node [vertex, fill = black] (1) at (-1,0) {};
        \node [vertex, fill = black] (2) at (-1,2) {};
        \node [vertex] (3) at (1,2) {};
        \node [vertex]  (4) at (1,0) {};
        
        \foreach \from/\to in {1/2, 2/3, 3/4, 4/1, 1/3, 2/4} 
        \draw [edge] (\from) to (\to);

    \end{scope}
    \begin{scope}[xshift=5cm]
        \node [vertex, fill = black] (1) at (-1,0) {};
        \node [vertex, fill = black] (2) at (-1,2) {};
        \node [vertex, fill = black] (3) at (1,2) {};
        \node [vertex, fill = black]  (4) at (1,0) {};
        
        \foreach \from/\to in {1/2, 2/3, 3/4, 4/1, 1/3, 2/4} 
        \draw [edge] (\from) to (\to);

    \end{scope}
    \end{tikzpicture}
\end{center}
Consider the $4$-uniform hypergraph $L$ with vertex set $\{0,1\}$
and symmetric edges $0001$ and $0111$. Now, on a given input 
graph $G$ consider the $4$-uniform hypergraph $H_G$ with vertex
set $V(G)$ and an edge $xyzw$ if and only if $x,y,z,w$
induce a $4$-clique in $G$. It is not hard to notice that
$G$ is a yes-instance to the colouring problem described above
if and only if there is a homomorphism $H_G\to L$, i.e., if
$H_G\in \CSP(L)$. Some readers may have noticed that $\CSP(L)$
corresponds to systems of linear equations over $\mathbb Z_2$
where each equation is of the form $x_1 + x_2 + x_3 + x_4 =1$:
each edge $xyzw$ of an input $H$ of $\CSP(L)$ corresponds to an
equation $x + y + z + w = 1$. Thus, we can solve $\CSP(L)$
(and the colouring problem above) by solving a system of linear equations
over $\mathbb Z_2$. Actually, this corresponds
to one of the four tractable boolean-domain CSPs
according to Schaefer's Theorem~\cite{Schaefer}. 

The technique used above is sometimes referred to as \emph{reduction to
the finite} (see, e.g.,~\cite[Section 10.5.3]{Book}). This technique is in 
particular useful for solving  infinite-domain CSPs expressible in
\emph{Monotone Monadic Strict NP} (MMSNP)~\cite{MMSNP-Journal}.
In short, the CSP of a $\tau$-structure $A$ is in MMSNP 
if there is a finite set of $k$-vertex-coloured $\tau$-structures
$\mathcal F$ such that a  finite $\tau$-structure $B$ maps homomorphically
to $A$ if and only if there is a  $k$-vertex colouring $B'$ of $B$ for which there
is no homomorphism $F\to B'$ for any $F\in \mathcal F$. So the graph colouring problem
described above is a CSP expressible in MMSNP. Actually, such reductions
to the finite were used to prove that infinite-domain CSPs expressible in MMSNP
exhibit  a P versus NP-complete dichotomy~\cite{FederVardi,MMSNP-Journal}.

Going back to the sandwich problems of split graphs, first recall that this
SP is the CSP of the complete $2$-edge-coloured graph $S^\ast$
consisting of a pair of countable sets of  vertices $X,Y$
where $X$ induces a red complete graph, $Y$ induces a complete blue
graph, and the edges between vertices in $X$ and vertices in $Y$
induce the random bipartite graph.
Let  $\mathcal F$ be the set of $2$-vertex-coloured $2$-edge-coloured graph
consisting of a blue edge with endpoints coloured with $0$, a red edge
with endpoints coloured with $1$, and both $2$-vertex-colourings of a red
loop, and both of a blue loop. Clearly, 
$\CSP(S^\ast)$ corresponds to the class of $2$-edge-coloured graphs
that admit a $2$-vertex-colouring avoiding all coloured graphs
from $\mathcal F$: vertices coloured with $0$ correspond to vertices
mapped to $X$, and vertices coloured with $1$ correspond to vertices
mapped to $Y$. The natural reduction to the finite for this CSP
reduces the SP of split graphs to the CSP of the following finite $2$-edge-coloured
graph with solid blue edges, and dashed red edges.
\begin{center}
    \begin{tikzpicture}[scale = 0.6]
    \begin{scope}[xshift=-5cm]
        \node [vertex] (1) at (-1,0) {};
        \node [vertex]  (2) at (1,0) {};

        \draw [edge, red, dashed] (1) to [bend left] (2);
        \draw [edge, red, dashed] (1) to [out=55,in=125,looseness=12] (1);
        \draw [edge, blue] (2) to [bend left] (1);
        \draw [edge, blue] (2) to [out=55,in=125,looseness=12] (2);
        \node (L1) at (0,-1) {$K$};
        
    \end{scope}
    \end{tikzpicture}
\end{center}
The CSP of this boolean structure $K$ encodes 
the restriction of 2-SAT where the inputs satisfy
that all clauses contain two positive literals or two negative literals:
on an instance $\phi:=(x\lor y) \land (\lnot z \lor \lnot w)\land \dots$ with variables $V$
construct a $2$-edge-coloured graph $G$ with vertex set $V$, and
there is a blue edge $uv$ if $u$ and $v$  appear positively in a clause $(v\lor u)$,
and a red edge $uv$ if $u$ and $v$ that appear
negatively in a clause $(\lnot v \lor \lnot u)$; there is a homomorphism $(V,B,R)\to K$
if and only if $\phi$ has a satisfiable assignment.

Finally, notice that this coincides with the reduction to 2-SAT considered
by Golumbic, Kaplan and Shamir~\cite{golumbicJA19}: on an input $(V,E,N)$ to the
sandwich problem of split graphs, they construct a 2-SAT instance with variables
$V$. For every blue edge $uv$ they include a clause $(x_u \lor x_v)$ and for
every red edge $uv$ a clause $(\lnot x_u\lor \lnot x_v)$.

\subsection{Threshold Graphs and Polymorphisms}
\label{sect:threshold}
A \emph{polymorphism} of a structure $A$ is a homomorphism $f\colon A^k\to A$,
where $k$ is the \emph{arity} of the polymorphism. The set of all polymorphisms 
of a structure $A$, denoted by $\Pol(A)$, forms a \emph{clone} (the precise definition is
not needed for this work). The algebraically  approach to CSPs arises from the observation
that if the polymorphism clone of a finite structure $A$ maps (in some algebraic
meaningful sense) to the polymorphism clone of a finite structure $B$, then $A$
pp-constructs $B$, and thus $\CSP(B)$ reduces in polynomial-time to $\CSP(A)$ (see, e.g.,~\cite{wonderland}).

A \emph{$4$-ary Siggers polymorphism} is a polymorphism $f\colon A^4\to A$ such that
$f(a,r,e,a) = f(r,a,r,e)$ for all $a,e,r \in A$. 
We now state the finite-domain dichotomy theorem~\cite{Zhuk20} (announced in~\cite{BulatovFVConjecture} and, independently, in~\cite{ZhukFVConjecture}).

\begin{theorem}\label{thm:dicho}
If a finite structure $A$ has a $4$-ary Siggers polymorphism, then
$\CSP(A)$ is polynomial time solvable; otherwise, the polymorphism clone
of $A$ maps to the polymorphism clone of $K_3$, and thus $\CSP(A)$ is $\NP$-complete.
\end{theorem} 

Before the finite-domain dichotomy was proved, there were other families of 
polymorphisms which implied that the corresponding CSP could be solved in polynomial
time. A \emph{totally symmetric} polymorphism is a polymorphism
$f\colon A^k\to A$ such that $f(x_1,\dots, x_k) = f(y_1,\dots, y_k)$ whenever
$\{x_1,\dots, x_k\} = \{y_1,\dots, y_k\}$. If a finite structure
$A$ has totally symmetric polymorphisms of all arities, then $\CSP(A)$ can 
be solved in polynomial time --- specifically, by a simple algorithm
called the \emph{arc-consistency procedure}~\cite{DalmauPearson}.
This tractability condition is known to fail in the infinite-domain setting in general (see~\cite[Section 13.1]{Book}).
However, in~\cite{BodMacpheTha}, the authors show that if $A$ has a first-order \emph{interpretation}
over $(\mathbb Q, <)$, and $A$ has totally symmetric polymorphisms of all arities,
then $\CSP(A)$ can be solved in polynomial time.

For this, we consider the following restricted case of first-order interpretations. 
A graph $G$ has a \emph{first-order  interpretation} over $(\mathbb Q, <)$
of \emph{dimension} $d\in \mathbb Z^+$, if there are first-order formulas $\delta_V(x_1,\dots, x_d)$
and $\delta_E(x_1,\dots, x_d, y_1,\dots, y_d)$ such that
\begin{itemize}
    \item the vertex set of $G$ is the set of tuples $(q_1,\dots, q_k)\in \mathbb Q^d$
    such that $\mathbb Q\models \delta_V(q_1,\dots, q_k)$, and
    \item there is an edge $(v_1,\dots, v_k)(u_1,\dots, u_k)$ if and only
    if $\mathbb Q\models \delta_E(v_1,\dots, v_d,u_1,\dots, u_d)$. 
\end{itemize}
For instance, consider the graph $T$ with vertex set $\{(x,y)\in \mathbb Q^2\colon x \neq y\}$, 
and where $(x,y)(z,w)$ forms an edge if and only if $(x,y) <_{lex} (z,y)$ and $x < y$
or if $(z,w) <_{lex} (x,y)$ and $z < w$, where $<_{lex}$ denotes the lexicographical
order, i.e., $(x,y) <_{lex} (z,w)$ if and only if $x < z$ or $x = z$ and $y < w$. 
It follows from the definition of $T$ that it has a first-order interpretation over $(\mathbb Q, <)$. 
We now see that $T$ is a universal  \emph{threshold graph}.

Threshold graphs have several equivalent definitions (see, e.g.,~\cite{mahaved1995}); 
we consider the most useful for this example. A graph $G$ is a threshold graph
if it can be obtained from $K_1$ by iteratively adding an independent vertex 
or a universal vertex. Note that every finite induced subgraph of 
$T$ is a threshold graph.  Indeed, let $(u_1,v_1),\dots, (u_n,v_n)$
be the vertices
of a finite induced subgraph $G_n$ of  $T$, 
and assume that $(u_i,v_i) <_{lex} (u_j,v_j)$ for every $1\leq i < j \le n$. 
Notice that if $u_i < v_i$, then $(u_i,v_i)$ is a universal vertex
in the graph $G_i$ with vertices $(u_i,v_i),\dots, (u_n,v_n)$, and if $v_i < u_i$,
then it is an isolated vertex in $G_i$. Therefore, $G_n$ can be obtained
from $K_1$ by iteratively adding isolated or universal vertices. On the other
hand, it is straightforward to use an inductive argument to embed
each finite threshold graph into $T$.

We now observe that threshold graphs are preserved by split blow-ups.
Let $G$ be a threshold graph with vertex set $v_1,\dots, v_n$ where for each $i\in[n]$
the vertex $v_i$ is either a universal or an isolated vertex in the subgraph
$G_i$ induced by $\{v_i,\dots, v_n\}$. Clearly, 
if $I :=\{v_i\in V\colon v_i$ is an isolated vertex in $G_i\}\cup\{v_n\}$,
then any split blow-up of $G$ with respect to $(I, V\setminus I)$ is a threshold graph. 
It thus follows from Lemma~\ref{lem:H->H*} that the SP for the class of threshold graphs 
is the CSP of $T^\ast$. We leave to the reader verifying that
the $k$-ary function $f\colon V(T)^k\to V(T)$ mapping a tuple $(x_1,\dots, x_k)$
to the lexicographical minimum of $\{x_1,\dots, x_k\}$ is indeed a polymorphism
$f\colon (T^\ast)^k\to T^\ast$ (the fact that it is totally symmetric follows
from the definition of $f$). Hence, it follows
from~\cite[Theorem 2.5 + Lemma 3.11]{BodMacpheTha} that the SP for threshold graphs
can be solved in polynomial time. 

\subsection{Comparability Graphs and the Bodirsky--Pinsker Conjecture}
\label{sect:comparability}
A relational structure $A$ is \emph{homogeneous} if for all finite
induced substructures $B_1,B_2$ of $A$ and every isomorphism
$f\colon B_1\to B_2$, there is an automorphism $g\colon A\to A$
 that agrees with $f$ on the vertices of $B_1$. For instance,
every complete graph $K_n$ is a homogeneous graph, but the path
on four vertices $v_1v_2v_3v_4$ is not: the mapping $f$ defined by
$v_1\mapsto v_1$ and $v_3\mapsto v_4$ is an isomorphism between the
substructures with vertices $\{v_1,v_3\}$ and $\{v_1,v_4\}$, but there
is no automorphism $g\colon P_4\to P_4$ that agrees with $f$ on $\{v_1,v_3\}$. 
Actually, finite and infinite homogeneous graphs have been fully 
classified~\cite{Gardiner,LachlanWoodrow};  in particular, the Random Graph is a homogeneous
graph. A prototypical example of a homogeneous digraph is $(\mathbb Q, <)$, i.e., 
the digraph where $x\to y$ if and only if $x < y$.

A structure $A$ is \emph{finitely bounded} if there is a finite
set of structures $\mathcal F$ such that the age of $A$ is the
class of $\mathcal F$-free structures. In graph theoretic terms, 
a graph $G$ is finitely bounded if its age has finitely many minimal
obstructions. For instance, the digraph $(\mathbb Q, <)$ is finitely
bounded: the minimal obstructions of the age of $(\mathbb Q, <)$
are a loop, a pair of independent vertices $2K_1$, the symmetric edge
$K_2$, and the directed triangle $\vec{C}_3$. So $\mathbb Q, <)$ is
a finitely bounded homogeneous structure. 
The Bodirsky--Pinsker Conjecture asserts that if
$B$ is first-order definable in a finitely bounded homogeneous structure
$A$, and $B$ does not pp-construct $K_3$, then $\CSP(B)$ is polynomial-time
solvable (see~\cite{BPP-projective-homomorphisms}; the presented formulation
is equivalent to the given formulation by a result from~\cite{BKOPP}; also
see~\cite[Conjectures 3.1 and 4.1]{Book}). As mentioned above, this conjecture has been
settled for the case when $B$ is first-order definable in $(\mathbb Q, <)$,
i.e., for temporal CSPs. 

The \emph{random poset} is the unique countable poset $P$ (up to isomorphism)
such that every countable poset embeds into $P$~\cite{Poset-Reducts} ---
we think of a poset as a transitive acyclic digraph. This digraph
is homogeneous and finitely bounded, and the Bodirsky--Pinsker Conjecture has
also been settled for structures that admit a first-order definition in $P$~\cite{posetCSP16,posetCSP18}.
We now see that the SP problem for comparability graphs is the CSP
of such a structure. 

A graph $G$ is a \emph{comparability graph} if there is a partial order $<$ on $V$
such that $uv\in E$ if and only if $u < v$ or $v < u$. So if $P_U$ is the underlying
graph of $P$, then $P_U$ is a universal comparability graph. Notice that
$P_U$  and $P_U^\ast$ are first-order definable in $P$.
Thus, the complexity classification of the SP for comparability 
graphs follows from Kompatscher and Pham~\cite{posetCSP16,posetCSP18} classification
of CSPs of reducts of the random poset. 
In particular, they showed that $\CSP(P_U^\ast)$ is $\NP$-complete~\cite[Example 3]{posetCSP16}.

\subsection{Generalized Split Graphs and PP-constructing Finite Structures}

\label{sect:gensplit}
For positive integers $p,q$ we say that a graph $G$ is a \emph{$(p,q)$-split graph}
if its vertex set $V$ can be partitioned into a pair of sets $(A,B)$ (where one of them
might be the empty set) such that $G[A]$ has independence number at most $p$, and
$G[B]$ has clique number at most $q$, i.e., $G[A]$ is a $(p+1)K_1$-free graph
and $G[B]$ is $K_{q+1}$-free graph. 
For instance, a graph $G$ is a $(1,1)$-split graph if and only if $G$ is a split graph, and it is a 
$(1,2)$-split graph if and only if it admits a partition into a clique and a triangle-free graph.

It follows from~\cite{golumbicJA19} and Section~\ref{sect:split} that 
if $p = q =1$, then the SP for $(1,1)$-split graphs can be solved in polynomial
time by reducing this problem to a tractable finite-domain CSP. Dantas,
Figueiredo, and Faria showed that the SP for $(p,q)$-split graphs is $\NP$-complete
whenever $p+q >2$~\cite{dantasDAM143}. Here, we reprove this  result to illustrate
pp-constructions of finite structures as a tool to obtain hardness results.
Specifically,  we show that if $p + q >2$, then there is a pp-construction of
$\text{1-IN-3}$ in a $2$-edge-coloured graph whose CSP describes the SP
for $(p,q)$-split graphs. We denote by $\text{1-IN-3}$ the boolean structure
encoding positive 1-IN-3 SAT, i.e.,  $\text{1-IN-3}:= (\{0,1\},\{(0,0,1),(0,1,0),(1,0,0)\})$.

Let $S$ be the graph obtained from the disjoint union $H_3 + I_\omega$, where $H_3$ is
the  homogeneous triangle-free graph, and $I_\omega$ is the countable independent set, 
and we also add edges $uv$ for $u\in V(H_3)$ and $v\in I_\omega$ independently with probability $1/2$.
So $S$ is a universal $(1,2)$-split graph. Since 
the class of $(1,2)$-split graphs is closed under split blow-ups, it follows that
$\CSP(S^\ast)$ is the SP for $(1,2)$-split graphs.

\begin{lemma}\label{lem:12-split}
    The $2$-edge-coloured graph $S^\ast$ pp-constructs $\text{1-IN-3}$.
\end{lemma}
\begin{proof}
    Consider the equivalence relation $x\sim y$ on $V(S^\ast)$
    defined as $V(I_\omega)^2 \cup V(H_3)^2$. 
    Consider also the ternary relation
    $R(x,y,z)$ pp-defined by $B(x,y)\land B(y,z)\land B(x,z)$.
    It is not hard to observe that $(V(S^\ast)/{\sim}, R)$ is isomorphic
    to 1-IN-3. Hence, it suffices to show
    that the equivalence relation $\sim$ is pp-definable in $S^\ast$. 
    We claim that $\sim$ is pp-definable by the following formula (see Figure~\ref{fig:formula})
    \begin{align*}
    \delta_\sim(x,y):= \exists z,w_1,w_2,w_3\; & \big( R(x,z)\land R(z,y)\land B(x,w_1) \land B(x,w_2)\land B(y,w_2)\land B(y,w_3)  \\
    \land &  B(w_1,w_2) \land B(w_2, w_3) \land B(z,w_1) \land B(z,w_2) \land B(z,w_3) \big ). \qedhere
    \end{align*}
\end{proof}

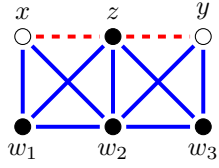
\begin{figure}
        \begin{center}
            \begin{tikzpicture}[scale = 0.8]
            \begin{scope}[xshift=5cm]
                \node [vertex, label = above:{$x$}] (x) at (-1.5,1.5) {};
                \node [vertex, fill = black, label = above:{$z$}] (z) at (0,1.5) {};
                \node [vertex, label = above:{$y$}] (y) at (1.5,1.5) {};

                \node [vertex, fill = black, label = below:{$w_1$}] (w1) at (-1.5,0) {};
                \node [vertex, fill = black, label = below:{$w_2$}] (w2) at (0,0) {};
                \node [vertex, fill = black, label = below:{$w_3$}]  (w3) at (1.5,0) {};

        \foreach \from/\to in {x/w1, x/w2, z/w1, z/w2, z/w3, y/w2, y/w3, w1/w2, w2/w3} 
        \draw [edge, blue] (\from) to (\to);

        \foreach \from/\to in {x/z, z/y} 
        \draw [edge, red, dashed] (\from) to (\to);
    \end{scope}
    \end{tikzpicture}
\end{center}
\caption{A picture for the pp $\delta_{\sim}$ from the proof of Lemma~\ref{lem:12-split}, where the filled vertices
    represent existentially quantified variables.}
    \label{fig:formula}
\end{figure}

\begin{theorem}\label{thm:pq-split}
    For every pair of positive integers $p,q$ one of the following statements holds.
    \begin{itemize}
        \item Either $p,q\le 1$, and in this case the sandwich problem for $(p,q)$-split graphs
        is polynomial-time solvable, or
        \item $p\ge 2$ or $q\ge 2$, and the sandwich problem for $(p,q)$-split graphs
        is $\NP$-complete.
    \end{itemize}
\end{theorem}
\begin{proof}
    The first item holds because the SP for split graphs can be solved in polynomial
    time~\cite{golumbicJA19} (see also Section~\ref{sect:split}).
    Since positive 1-IN-3 SAT is $\NP$-complete, it follows via Lemma~\ref{lem:12-split} and
    Lemma~\ref{lem:pp-construction-reduction} that $\CSP(S^\ast)$ is $\NP$-complete, and so
    the SP for $(1,2)$-split graphs
    is $\NP$-complete (Lemma~\ref{lem:H->H*}). By reducing to the complement, 
    we see that the SP for $(2,1)$-split graphs is $\NP$-complete. Finally, we
    show that the SP for $(p,q)$-split graphs reduced in polynomial time
    to the SP for $(p,q+1)$-split graphs (and symmetrically, to the SP for
    $(p+1,q)$-split graphs). Consider an input $(V,E,N)$ to the SP
    for $(p,q)$-split graphs, and let $v_1,\dots, v_{p+1}$ be new vertices
    not in $V$. Let $V_1 = V\cup \{v_1,\dots, v_{p+1}\}$,
    $E_1 = E\cup \{uv_i\colon u\in V, 1\le i\le p+1\}$, and $N_1 = N\cup \{v_iv_j\colon i\neq j, 1\leq i,j\le p+1\}$. 
    It is straightforward to observe that if $(V,E')$ is a $(p,q)$-split graph
    with $E\subseteq E'$ and $E'\cap N= \varnothing$, then $(V_1, E_1\cup E')$ is
    a $(p,q+1)$-split graph with $E_1\subseteq E_1\cup E'$ and $(E_1\cup E')\cap N_1 = \varnothing$.
    In other words, if $(V,E,N)$ is a yes-instance, then $(V_1,E_1,N_1)$ is a yes-instance. 
    Conversely, if $(V_1,E')$ is a $(p,q+1)$-split graph such that $E_1\subseteq E'$
    and $E'\cap N_1=\varnothing$, then $(V, E'\cap V^2)$ is a $(p,q)$-split graph
    with $E\subseteq E'\cap V^2$ and $( E'\cap V^2)\cap N = \varnothing$. 
    Therefore, the reduction $(V,E,N)\mapsto (V_1,E_1,N_1)$ is a polynomial-time
    reduction from the SP for $(p,q)$-graphs to the SP for $(p,q+1)$-graphs. 
    One can argue analogously that the SP for $(p,q)$-split graphs reduces in
    polynomial-time to the SP for $(p+1,q)$-split graphs, and since the SPs for $(1,2)$-split
    graphs and for $(2,1)$-split graphs are $\NP$-complete, we conclude that 
    the SP for $(p,q)$-split graphs is $\NP$-complete whenever $p\ge 2$
    or $q\ge 2$.
\end{proof}

\subsection{Permutation Graphs and PP-Constructing Infinite Structures}
\label{sect:permutation}
A graph $G$ is a \emph{permutation graph} if its vertices can be represented
as straight line segments between two parallel lines, and two vertices are adjacent
if and only if the corresponding line segments intersect. One may further require
that no two different line segments intersect in their end-points. However, 
for every finite amount of line segments, we can apply a small enough perturbation to their
end-points so that no two different line segments have a common end-point, and 
two line segments intersect in the perturbed configuration if and only if they intersect in 
the original configuration. Hence, both definitions of permutation graphs are equivalent for finite graphs,
and we choose to allow for different line segments to have a common end-point. 
One can easily construct a universal permutation graph $P$: the vertex set
of $P$  is $\mathbb Q^2$, and there is an edge $(x_1,y_1)(x_2,y_2)$ if $x_1 \le x_2$ and $y_2 \le y_1$,
or if $x_2 \le x_1$ and $y_1 \le y_2$. So the blue neighborhood of a vertex $v$ in $P^\ast$
corresponds to the closed bottom right and top left quadrants defined by $v$, and the red neighbors
to the open bottom left and top right quadrants defined by $v$ --- we illustrate this in
Figure~\ref{fig:permutation}.

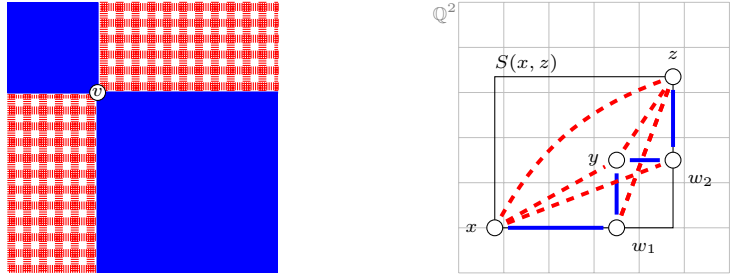
\begin{figure}[ht!]
\centering
\begin{tikzpicture}

  \begin{scope}[xshift = -6cm, scale = 0.6]
    \fill[blue] (-3,1) rectangle (-1,3);
    \fill[blue] (-1,-3) rectangle (3,1);
    \draw[step=0.05cm, red, dashed] (-0.95,1.05) grid (3,3);
    \draw[step=0.05cm, red, dashed] (-3,-3) grid (-1.05,0.95);
    \draw[-, thick, blue] (-3,1) -- (3,1);
    \draw[-, thick, blue] (-1,-3) -- (-1,3);
    \node [vertex] (v) at (-1,1) {\scriptsize $v$};

  \end{scope}

  \begin{scope}[scale = 0.6]    
    \draw[step=1cm, lightgray, very thin] (-3,-3) grid (3,3);
    \draw (-2.2,-2) rectangle (1.75,1.35);
    
    \node [vertex, label = left:{\scriptsize $x$}] (1) at (-2.2,-2) {};
    \node [vertex, label = 315:{\scriptsize $w_1$}] (2) at (0.5,-2) {};
    \node [vertex, label = left:{\scriptsize $y$}] (3) at (0.5,-0.5) {};
    \node [vertex, label = above:{\scriptsize $z$}] (5) at (1.75,1.35) {};
    \node [vertex, label = 315:{\scriptsize $w_2$}] (4) at (1.75,-0.5) {};
    \node at (-1.5,1.65) {\scriptsize $S(x,z)$};
    
    \node at (-3.3,2.75) {\scriptsize \color{gray} $\mathbb Q^2$};
      
    \foreach \from/\to in {1/2, 2/3, 3/4, 4/5}     
    \draw [edge, blue] (\from) to  (\to);
    \foreach \from/\to in {1/3, 2/5, 2/5, 1/4, 3/5}     
    \draw [edge, red, dashed] (\from) to  (\to);
    \draw [edge, red, dashed] (1) to [bend left = 20]  (5);
  \end{scope}

\end{tikzpicture}
\caption{To the left, a depiction of the blue (solid) and red (dashed) neighbourhoods
of a vertex $v$ in the CSP template $(\mathbb Q^2, B, R)$ that describes the SP
for permutation graphs. To the right, a depiction of  witnesses $w_1,w_2$
$(x,y,z)\in R$ whenever $x$ and $z$ are the bottom left and top right
corners of the rectangle $S(x,z)$ and $y$ lies inside $S(x,z)$.
}
\label{fig:permutation}
\end{figure}

The hardness proof of the SP for permutation graphs from~\cite{golumbicJA19}
corresponds to the following pp-construction of betweenness in $P^\ast$.
We first pp-define a ternary relation $R(x,y,z)$ over $P^\ast = (\mathbb Q^2, E, N)$,
and we then see that $(\mathbb Q^2, R)$ is homomorphically equivalent to $(\mathbb Q, \Betw)$.
Consider the following definition. 
\begin{align*}
    R(x,y,z):=  \exists & w_1,w_2 \big (B(x,w_1) \land B(w_1,y) \land B(y,w_2) \land B(w_2,z)\\
            \land & R(x,y) \land R(x,w_2) \land R(x,z) \land R(w_1,w_2) \land R(w_1,z) \land R(y,z) \big ) 
\end{align*}
We first propose the following geometric interpretation of the ternary relation $R$. 
Any pair of points $x,y\in \mathbb Q^2$ that do not belong to the same horizontal
or the same vertical line defines a unique rectangle $S$ where $x$ and $y$ are opposite
corners of the rectangle. Notice that if $x$ and $y$ correspond to the
top left and bottom right corners of $S(x,y)$, then any vertex inside $S$ is 
connected to both $x$ and $y$ by a blue edge in $P^\ast$; if $x$ and $y$ are
the bottom left and top right corner of $S(x,y)$, then any vertex inside $S$ is
connected to $x$ and to $y$ by a red edge. We claim that $(x,y,z) \in R$
if and only if $x$ and $z$ are the bottom left and top right corners of $S(x,z)$
and $y$ lies inside $S(x,z)$. Indeed, suppose that the latter holds. For each $a\in \{x,y,z\}$ let $a = (a_1,a_2)$. If $x$ is the bottom left
corner of $S(x,z)$ and $z$ the top right one, then the points $w_1 = (y_1,x_2)$
and $w_2 = (z_1,y_2)$ witness that $(x,y,z)$ belongs to $R$ (see  Figure~\ref{fig:permutation})
--- the case where $x$ is the top right corner of $S(x,z)$ and $z$ the bottom left is 
symmetric. We prove the contrapositive of the remaining implication. Assume 
that the two vertices $a,b\in \{x,y,z\}$ are the top right and bottom left
corners of $S(a,b)$. Then $c \in \{x,y,z\}\setminus \{a,b\}$ lies inside $S(a,b)$, because 
otherwise $x,y,z$ do not induce a red triangle, and hence $(x,y,z)\not\in R$. 
Since $y$ does not belong to the square $S(x,z)$, we may assume without loss of generality
that $c = z$ (the case where $c = x$ can be shown analogously). In this case, notice
that any neighbor $w_1$ of $x$ and of $y$ is also a neighbor of $z$, and so
$(x,y,z)\not\in R$. 

Using the geometric description of $R$ it follows that
the function $f\colon \mathbb Q\to \mathbb Q^2$ mapping $\mathbb Q$ to
the diagonal by $q\mapsto (q,q)$ defines an embedding of $(\mathbb Q,\Betw)$ into
$(\mathbb Q^2, R)$. In particular, $(\mathbb Q, \Betw)$ is isomorphic
to the substructure $(\Delta, R_{|\Delta})$ of $(\mathbb Q^2, R)$ induced
by the diagonal $\Delta = \{(q,q)\in \mathbb Q^2\}$. Using the same
geometric description of $R$ it also follows that the orthogonal
projection $h\colon \mathbb Q^2\to \Delta$ of $\mathbb Q^2$ onto the
diagonal $\Delta$ defines a homomorphism
$h\colon (\mathbb Q^2, R)\to (\Delta, R_{|\Delta})\cong (\mathbb Q, \Betw)$.
Hence, $(\mathbb Q^2, R)$ is homomorphically equivalent to 
$(\mathbb Q, \Betw)$, and since $R$ has a primitive positive definition
on $P^\ast$, we conclude that $P^\ast$ pp-constructs
$(\mathbb Q, \Betw)$. This implies that the SP for permutation graphs
is NP-complete. 

We point out that the gadget reduction from the betweenness problem
to the SP for permutation graphs proposed in~\cite{golumbicJA19} corresponds
to the gadget obtained from the pp-construction above via Lemma~\ref{lem:pp-construction-reduction}
and Proposition~\ref{prop:SP->CSP}.

\section{Line graphs}
\label{sect:line}
Given a graph $G$, its \emph{line graph} $L(G)$ is the graph whose vertex set consists of the edges of $G$, and where two vertices
are adjacent if and only if the corresponding edges are incident to a common vertex in $G$.
A graph $H$ is \emph{a line graph} if there is a graph $G$ such that $H$ is isomorphic to the line graph of $G$. 
Unfortunately (for the CSP approach to sandwich problems), 
the class of line graphs is not  
closed under split blow-ups. However, line graphs of multigraphs are closed under co-blow-ups;
this class has 
been considered in the literature and admits 
nice structural characterizations~\cite[Theorem 4.1]{petersonDAM126}.
Similarly, the class of line graphs of triangle-free multigraphs is closed under co-blow-ups as well, and corresponds to the class
of \emph{dominoes} studied in~\cite{kloksGTCS1995,metelskySIDMA16}.
In this section we classify the complexity of 
the SP for line graphs of multigraphs (Corollary~\ref{cor:line-multigraphs}),
and for line graphs of  bipartite multigraphs (Proposition~\ref{prop:line-multi-bip}).

\subsection{Line graphs of multigraphs}
The \emph{Johnson graph} $J(2)$ is the line graph of the infinite clique --- the Johnson graph $J(k)$
is defined for every positive integer $k$ (see, e.g.~\cite{Thomas96}), but its definition
is not needed for the present work. Observe that the Johnson graph $J(2)$ is a
universal 
line graph. So if $J(2)'$ is obtained from $J(2)$ by substituting each
vertex of $J(2)$ by a countable set of co-twins, then $J(2)'$ is a universal line
graph of multigraphs. It now follows via Lemma~\ref{lem:H->H*} that $\CSP((J(2)')^\ast)$ describes
the SP for the class of line graphs of multigraphs (which are simple graphs). 

We will use the following simple lemma a couple of times in this section. We say that a
graph $H$ is \emph{point-incomparable} if for every pair of vertices $x,y\in V$ with
$x\neq y$, there are $w_1,w_2\in V\setminus \{x,y\}$ such that
$w_1x, w_2y\in E$ and $w_1y, w_2x\not\in E$. 

\begin{lemma}\label{lem:H'-pp-H}
    Let $H$ be a graph and let $H'$ be a split blow-up of $H$. If $H$ is a point-incomparable graph,
    then $(H')^\ast$ pp-constructs $H^\ast$, and so $\CSP(H^\ast)$ reduces in polynomial-time to
    $\CSP((H')^\ast)$. 
\end{lemma}
\begin{proof}
    Let $\sim$ be the equivalence relation on $V(H')$ defined by the blow-up equivalence, i.e.,
    $x\sim y$ if and only if there is a vertex $v\in V(H)$ such that $x$ and $y$ belong to the
    blow-up or co-blow-up of $v$. Since $H$ is point-distinguishing, it neither has 
    a pair of twins nor of co-twins. Hence, 
    $x\sim y$ holds if and only if $x$ and $y$ are twins or co-twins in $H'$.
    Now, consider the following pp-definitions.
    \begin{align*}
        \delta_R(x,y) := R(x,y) \land \exists w_1,w_2\; \big (B(x,w_1) \land B(y,w_2)\land R(x,w_2)\land R(y,w_1) \big)\\
        \delta_B(x,y) := B(x,y) \land \exists w_1,w_2\; \big (B(x,w_1) \land B(y,w_2)\land R(x,w_2)\land R(y,w_1) \big )
    \end{align*}
    On a high level, the $2$-edge-coloured graph $H^\delta:=(V(H'), \delta_B,\delta_R)$ is the
    $2$-edge-coloured graph obtained from $(H')^\ast$ be removing all edges inside each $\sim$-equivalence
    class. Hence, the quotient $H^\delta/{\sim}$ is isomorphic to $H^\ast$, and it readily
    follows that $H^\ast$ and $H^\delta$ are homomorphically equivalent. Hence, $(H')^\ast$ pp-constructs
    $H^\ast$. The fact that  $\CSP(H^\ast)$ reduces in polynomial-time to
    $\CSP((H')^\ast)$ follows via Lemma~\ref{lem:pp-construction-reduction}. 
\end{proof}

In terms of the CSP approach to graph sandwich problems, this lemma has the following application.

\begin{proposition}\label{prop:C'-C}
    Let $\mathcal C$ be a class of graphs, and let $\mathcal C'$ be the class of graphs
    obtained as co-blow-ups of graphs in $\mathcal C$. If $H$ is a point-incomparable
    universal graph in $\mathcal C$, and $\CSP(H^\ast)$ is $\NP$-hard, then the sandwich
    problem for $\mathcal C'$ is $\NP$-hard as well. 
\end{proposition}
\begin{proof}
    Let $H'$ be the graph obtained from $H$ by blowing up each vertex to a countable
    set of co-twins. It is straightforward to observe that $H'$ is a universal
    graph in $\mathcal C'$, and that $H'$ is a split blow-up of $H$. 
    Since $\mathcal C'$ is closed under co-blow-ups, it follows by Lemma~\ref{lem:H->H*}
    that $\CSP((H')^\ast))$ described the SP for $\mathcal C'$. By Lemma~\ref{lem:H'-pp-H} 
    it follows that $\CSP((H')^\ast)$ is $\NP$-hard because $\CSP(H^\ast)$ is $\NP$-hard and $H$ is point-incomparable, 
    and thus the SP for $\mathcal C'$ is $\NP$-hard. 
\end{proof}

Recognizing line graphs of multigraphs can be done in
polynomial time. In fact, 
there are only finitely many minimal obstructions to this class~\cite{bermondJMPA52,zverovichDMA7}. 
Thus, the SP for multigraphs is in NP. It is straightforward to observe that the line
graph $J(2)$ of the infinite clique is a point-incomparable graph, so by Proposition~\ref{prop:C'-C},
it suffices to show that the CSP of $J(2)^\ast$ is $\NP$-hard to conclude that
the SP for line graphs of multigraphs is NP-complete. We will build again on results
from the CSP literature.

A $\sigma$-structure  $B$ is a
\emph{first-order reduct} 
of a $\tau$-structure $A$ if 
\begin{itemize}
    \item the vertex set of $B$ equals the vertex set of $A$, and
    \item for every symbol $S\in \sigma$ of arity $s$, there is a first-order
    $\tau$-formula $\phi_S(x_1,\dots, x_s)$ with $s$ many free-variables,
    such that a tuple $(b_1,\dots, b_s)$ belongs to the interpretation
    of $S$ in $B$ if and only if $A\models \phi_S(b_1,\dots, b_s)$.
\end{itemize}
For instance, if $D$ is a digraph, and $G$ is obtained
from $D$ by forgetting the orientation of the edges, then 
$G$ is a first-order reduct of $D$: simply consider the
first-order formula $\phi_E(x,y):= E(x,y)\lor E(y,x)$.
 Also notice that $H^\ast$ is a first-order reduct of $H$.

Theorem 44 in~\cite{bodirskySw} shows a complexity classification of
the CSP of first-order reducts $R$ of $J(2)$. It states that 
$\CSP(R)$ is NP-hard unless 
there exists another structure $R'$ such that
\begin{itemize}
    \item $R$ and $R'$ are \emph{homomorphically equivalent}, i.e., there exists a homomorphism 
    from $R$ to $R'$ and vice versa;
    \item $R'$ is a \emph{model-complete core}, i.e., every \emph{endomorphism of $R'$} (i.e., every homomorphism from $R'$ to $R'$ preserves all first-order formulas;
    \item $R'$ just has one element, or $R'$ is preserved by all permutations, and there exists an injective homomorphism from $(R')^2$ to $R'$. 
\end{itemize}
Note that the final item is extremely restrictive. It is easy to see that for the first-order reduct $J(2)^*$ of $J(2)$ there does not exist an $R'$ as stated above: clearly, since both the edge relation $E$ and the relation $N$ in $J(2)^*$ are non-empty, and
$R'$ is homomorphically equivalent, then
the same is true for $R'$. Hence, it cannot consist of just one point, and for the same reason it cannot be preserved by all permutations. It thus follows by Theorem~44 in~\cite{bodirskySw}
that $\CSP(J(2)^\ast)$ is $\NP$-complete.

\begin{corollary}\label{cor:line-multigraphs}
    The sandwich problem for line graphs of multigraphs is $\NP$-complete.
\end{corollary}

\subsection{Line graphs of bipartite multigraphs}
It is well-known that a graph $G$ is the line graph of a bipartite graph if and only if its vertices can be
represented as points of the grid $\mathbb Z^2$ in such a way that a $uv\in E(G)$ if and
only if $u$ and $v$ are represented by points on the same horizontal or the same vertical
line~\cite{petersonDAM126}.  The \emph{grid graph} $\GR$ is universal line graph of bipartite graphs,
and defined as follows:
\begin{itemize}
    \item its vertex set is $V(\GR) =  \mathbb Z^2$, and 
    \item its edge set is $\{(a,b)(c,d)\mid a = b$
and $b \neq d$, or $a\neq b$ and $b = d\}$.
\end{itemize}
It has been noted recently~\cite[Proposition 15]{bokPREPATATION} that a graph $G$ is a line graph
of a bipartite graph if and only if there is a $2$-edge-colouring of $G$ that avoids the following
forbidden structures, where dashed red indicates non-adjacent pairs of vertices, and thick blue and
thin green edges represent different edge colours (which corresponds to vertical and horizontal edges
in the grid representation of a line graph of a bipartite graph). 
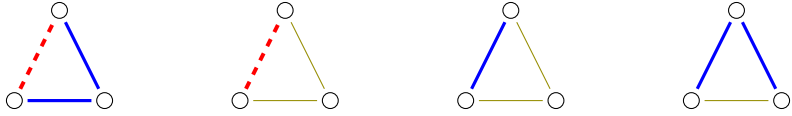
\begin{center}
    \begin{tikzpicture}[scale = 0.6]
    \begin{scope}[xshift=-5cm]
        \node [vertex] (1) at (-1,0) {};
        \node [vertex] (2) at (1,0) {};
        \node [vertex] (3) at (0,2) {};
        
        \draw [edge, red, dashed] (1) to (3);
        \draw [edge, blue, very thick] (2) to (3);
        \draw [edge, blue, very thick] (1) to (2);

    \end{scope}
    \begin{scope}
        \node [vertex] (1) at (-1,0) {};
        \node [vertex] (2) at (1,0) {};
        \node [vertex] (3) at (0,2) {};
        
        \draw [edge, red, dashed] (1) to (3);
        \draw [edge, olive, thin] (2) to (3);
        \draw [edge, olive, thin] (1) to (2);
       
    \end{scope}
    \begin{scope}[xshift=5cm]
        \node [vertex] (1) at (-1,0) {};
        \node [vertex] (2) at (1,0) {};
        \node [vertex] (3) at (0,2) {};
        
        \draw [edge, blue, very thick] (1) to (3);
        \draw [edge, olive, thin] (2) to (3);
        \draw [edge, olive, thin] (1) to (2);

    \end{scope}
    \begin{scope}[xshift=10cm]
        \node [vertex] (1) at (-1,0) {};
        \node [vertex] (2) at (1,0) {};
        \node [vertex] (3) at (0,2) {};
        
        \draw [edge, blue, very  thick] (1) to (3);
        \draw [edge, blue, very  thick] (2) to (3);
        \draw [edge, olive, thin] (1) to (2);
    \end{scope}
    \end{tikzpicture}
\end{center}
This implies that an input $(V,E,N)$ to the SP for line graphs
of bipartite graphs is a yes-instance if and only if its
set of unordered pairs of vertices $uv$ with $u\neq v$ admit
a $\{b,g,r\}$-colouring satisfying:
\begin{itemize}
    \item if $uv\in N$, then $uv$ is coloured with $r$ (red),
    \item if $uv\in E$, then $uv$ is coloured with $b$ (blue) or $g$ (green), and
    \item for any three different vertices $u,v,w$ the triangle of unordered
    pairs $uv,vw,wu$ is either monochromatic $ggg$, $bbb$, $rrr$, or
    of the form $rrb$, $rrg$, or $bgr$.
\end{itemize}
This characterization yields a natural reduction of the SP for 
line graphs of bipartite graphs to the CSP of a finite
structure $A$ --- this is another example of a reduction 
to the finite introduced in Section~\ref{sect:split}. 

The vertex set of $A$ is $\{b,g,r\}$, the signature
contains two unary predicates $U_N$ and $U_E$, and one ternary symmetric
predicate $T$. The symbol $U_N$ 
denotes the relation 
$\{r\}$ in $A$, 
the symbol 
$U_E$ denotes $\{b,g\}$, and the symbol $T$ denotes $\{bbb, ggg, rrr, brr, grr, bgr\}$ (recall that this is a symmetric
relation; also note that we list precisely those triples that are not shown in the
illustration of the 2-edges colourings shown above). 

\begin{lemma}\label{lem:line-graphs-A}
    The SP for line graphs of bipartite graphs reduces in polynomial-time to $\CSP(A)$. 
\end{lemma}
\begin{proof}
    Let $\tau:=\{U_N,U_E,T\}$ be the signature of $A$ defined above.
    On an input $(V,E,N)$ to the SP of line graphs of bipartite graphs
    consider the finite structure $I$ whose vertices are unordered pairs
    $uv$ for all $u,v\in V$ with $u\neq v$. For every three pairwise
    distinct unordered pairs $u_1v_1,u_2v_2,u_3v_3$ spanning exactly 
    three vertices, add the unordered triple $u_1v_1~u_2v_2~u_3v_3$ to 
    the relation $T(I)$. Finally, if $uv\in N$ colour the vertex $uv$ in $I$
    with $U_N$, and if $uv\in E$, colour $uv$ in $I$ with $U_E$.
    It follows from the arguments above that 
    $I\to A$ if and only if $(V,E,N)$ is a yes-instance of the SP
    for line graphs of bipartite graphs.
\end{proof}

As we will see below, the CSP of $A$ is $\NP$-complete, so Lemma~\ref{lem:line-graphs-A}
does not imply that the SP for line graphs of bipartite graphs is polynomial-time solvable. 
We now show that $\GR^\ast$ pp-constructs $A$, and since the class of
line graphs of bipartite multigraphs equals the class of co-blow-ups of line graphs
of bipartite graphs, we conclude via Proposition~\ref{prop:C'-C} that the SP for line graphs
of bipartite multigraphs is $\NP$-complete.

We first notice that the inequality relation $\neq$ has a primitive positive definition in $\GR^\ast$. This
follows from the fact that $\GR$ is point-incomparable, i.e., for any two distinct vertices $u$ and $v$ of
$\GR$ there is a vertex $w$ adjacent to $u$ and not adjacent to $v$, and vice versa. Hence, the formula
$\exists w (B(w,x)\land R(w,y))$ pp-defines inequality in $\GR^\ast$. Now, we consider the $4$-ary relation
$S(x_1,x_2,y_1,y_2)$ defined by ``either $x_1x_2$ and $y_1y_2$ are both vertical edges,
or both are horizontal edges, or $x_1,x_2$ and $y_1,y_2$ are pairs on non-equal non-adjacent
vertices''. To see that $S$ admits a primitive positive definition in $\GR^\ast$, we consider the pp-formula
associated to the following graph, where the dotted edges indicate the inequality relation, and
black filled vertices correspond to existentially quantified variables.
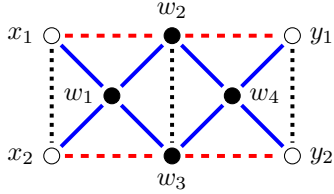
\begin{center}
    \begin{tikzpicture}[scale = 0.8]

    \begin{scope}[xshift=5cm]
        \node [vertex, label = left:{$x_1$}] (x1) at (-2,1) {};
        \node [vertex, label = left:{$x_2$}] (x2) at (-2,-1) {};
        \node [vertex, label = right:{$y_1$}] (y1) at (2,1) {};
        \node [vertex, label = right:{$y_2$}]  (y2) at (2,-1) {};

        \node [vertex, fill = black, label = left:{$w_1$}] (1) at (-1,0) {};
        \node [vertex, fill = black, label = above:{$w_2$}] (2) at (0,1) {};
        \node [vertex, fill = black, label = below:{$w_3$}]  (3) at (0,-1) {};
        \node [vertex, fill = black, label = right:{$w_4$}]  (4) at (1,0) {};

        \foreach \from/\to in {x1/1, x2/1, 1/2, 1/3, 3/4, 2/4, 4/y1, 4/y2} 
        \draw [edge, blue] (\from) to (\to);

        \foreach \from/\to in {x1/2, 2/y1, x2/3, 3/y2} 
        \draw [edge, red, dashed] (\from) to (\to);
        \draw [edge, dotted] (x1) to  (x2);
        \draw [edge, dotted] (y1) to  (y2);
        \draw [edge, dotted] (2) to  (3);
        
    \end{scope}
    \end{tikzpicture}
\end{center}
\begin{align*}
    \gamma(x_1,x_2,y_1,y_2): = \exists w_1,w_2,w_3,w_4\; \big ( & B(w_1,w_2)\land B(w_2,w_4)\land B(w_4,w_3)\land B(w_3,w_1)\\
    \land\; & B(x_1,w_1)\land B(x_2,w_1) \land R(x_1, w_2) \land R(x_2,w_3) \\
    \land\; & B(y_1,w_4)\land B(y_2,w_4) \land R(y_1, w_2) \land R(y_2,w_3) \\
    \land\; & x_1\neq x_2\land y_1\neq y_2 \land w_2\neq w_3 \big ).
\end{align*}
In Figure~\ref{fig:gamma} we present a depiction showing that if $x_1x_2$ and $y_1y_2$ are
both vertical edges, then $\GR^\ast\models \gamma(x_1,x_2,y_1,y_2)$, and similarly, when $x_1,x_2$ and 
$y_1,y_2$ are pairs of non-equal non adjacent vertices. It follows that 
if $(u_1,u_2,v_1,v_2)$ belongs to the relation $S$ defined above, then 
$\GR^\ast\models \gamma(u_1,u_2,v_1,v_2)$; the proof of the converse implication is straightforward. 

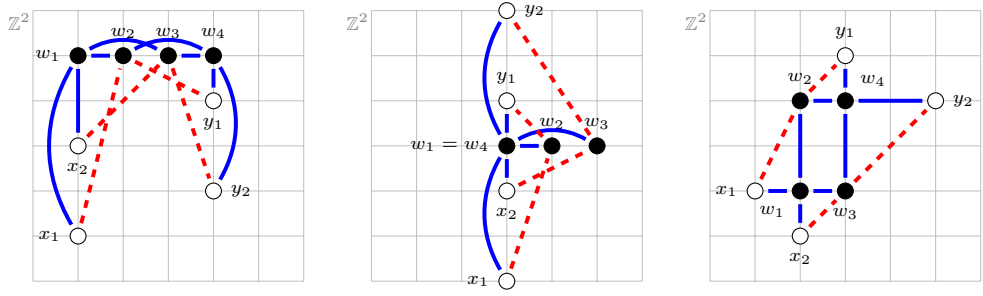
\begin{figure}[ht!]
\centering
\begin{tikzpicture}

  \begin{scope}[scale = 0.6]    
    \draw[step=1cm, lightgray, very thin] (-3,-3) grid (3,3);
    
    \node [vertex, label = left:{\scriptsize $x_1$}] (x1) at (-2,-2) {};
    \node [vertex, label = below:{\scriptsize $x_2$}] (x2) at (-2,0) {};
    \node [vertex, fill = black, label = left:{\scriptsize $w_1$}] (w1) at (-2,2) {};
    \node [vertex, fill = black, label = above:{\scriptsize $w_2$}] (w2) at (-1,2) {};
    \node [vertex, fill = black, label = above:{\scriptsize $w_3$}] (w3) at (0,2) {};
    \node [vertex,fill = black,  label = above:{\scriptsize $w_4$}] (w4) at (1,2) {};
    \node [vertex, label = below:{\scriptsize $y_1$}] (y1) at (1,1) {};
    \node [vertex, label = right:{\scriptsize $y_2$}] (y2) at (1,-1) {};
    
    \node at (-3.3,2.75) {\scriptsize \color{gray} $\mathbb Z^2$};
      
    \foreach \from/\to in  {x2/w1, w1/w2,  w3/w4, w4/y1} 
        \draw [edge, blue] (\from) to (\to);
        \draw [edge, blue] (w1) to [bend left] (w3);
        \draw [edge, blue] (w2) to [bend left] (w4);
        \draw [edge, blue] (x1) to [bend left] (w1);
        \draw [edge, blue] (y2) to [bend right] (w4);

        \foreach \from/\to in {x1/w2, w2/y1, x2/w3, w3/y2} 
        \draw [edge, red, dashed] (\from) to (\to);
  \end{scope}

   \begin{scope}[scale = 0.6, xshift = 7.5cm]    
    \draw[step=1cm, lightgray, very thin] (-3,-3) grid (3,3);
    
    \node [vertex, label = left:{\scriptsize $x_1$}] (x1) at (0,-3) {};
    \node [vertex, label = below:{\scriptsize $x_2$}] (x2) at (0,-1) {};
    \node [vertex, fill = black, label = left:{\scriptsize $w_1 = w_4$}] (w1) at (0,0) {};
    \node [vertex, fill = black, label = above:{\scriptsize $w_2$}] (w2) at (1,0) {};
    \node [vertex, fill = black, label = above:{\scriptsize $w_3$}] (w3) at (2,0) {};
    \node [vertex, label = above:{\scriptsize $y_1$}] (y1) at (0,1) {};
    \node [vertex, label = right:{\scriptsize $y_2$}] (y2) at (0,3) {};
    
    \node at (-3.3,2.75) {\scriptsize \color{gray} $\mathbb Z^2$};
      
    \foreach \from/\to in  {x2/w1, w1/w2,  w1/y1} 
        \draw [edge, blue] (\from) to (\to);
        \draw [edge, blue] (w1) to [bend left] (w3);
        \draw [edge, blue] (x1) to [bend left] (w1);
        \draw [edge, blue] (y2) to [bend right] (w1);

        \foreach \from/\to in {x1/w2, w2/y1, x2/w3, w3/y2} 
        \draw [edge, red, dashed] (\from) to (\to);
  \end{scope}

    \begin{scope}[scale = 0.6, xshift = 15cm]    
    \draw[step=1cm, lightgray, very thin] (-3,-3) grid (3,3);
    
    \node [vertex, label = left:{\scriptsize $x_1$}] (x1) at (-2,-1) {};
    \node [vertex, label = below:{\scriptsize $x_2$}] (x2) at (-1,-2) {};
    \node [vertex, fill = black, label = 225:{\scriptsize $w_1$}] (w1) at (-1,-1) {};
    \node [vertex, fill = black, label = above:{\scriptsize $w_2$}] (w2) at (-1,1) {};
    \node [vertex, fill = black, label = below:{\scriptsize $w_3$}] (w3) at (0,-1) {};
    \node [vertex,fill = black,  label = 45:{\scriptsize $w_4$}] (w4) at (0,1) {};
    \node [vertex, label = above:{\scriptsize $y_1$}] (y1) at (0,2) {};
    \node [vertex, label = right:{\scriptsize $y_2$}] (y2) at (2,1) {};
    
    \node at (-3.3,2.75) {\scriptsize \color{gray} $\mathbb Z^2$};
      
    \foreach \from/\to in  {x2/w1, w1/w2,  w3/w4, w4/y1, w1/w3, w2/w4, x1/w1, y2/w4} 
        \draw [edge, blue] (\from) to (\to);

        \foreach \from/\to in {x1/w2, w2/y1, x2/w3, w3/y2} 
        \draw [edge, red, dashed] (\from) to (\to);
  \end{scope}

\end{tikzpicture}
\caption{To the left, an illustration of witnesses for $w_1,w_2,w_3,w_4$ showing that
if the edges $x_1x_2$ and $x_1x_2$ are both vertical  laying in different vertical lines, then
$\GR^\ast\models \gamma(x_1,x_2,y_1,y_2)$. Similarly, in the middle an illustration that
$\GR^\ast\models \gamma(x_1,x_2,y_1,y_2)$ if $x_1x_2$ and $y_1y_2$ are vertical edges 
on the same vertical line. (By rotating the gadgets ninety degrees, we obtain similar pictures
showing that $\GR^\ast\models\gamma(x_1,x_2,y_1,y_2)$ when $x_1x_2$ and $y_1y_2$ are both
horizontal blue edges). To the right, a depiction showing that $\GR^\ast\models \gamma(x_1,x_2,y_1,y_2)$
if $x_1,x_2$ and $y_2,y_2$ are pairs of non-equal non-adjacent vertices.}
\label{fig:gamma}
\end{figure}

\begin{lemma}\label{lem:GR-pp-A}
    The $2$-edge-coloured graph $\GR^\ast$ pp-constructs the finite
    structure $A$.
\end{lemma}
\begin{proof}
    We consider a $2$-dimensional pp-construction. The domain
    formula $\delta_D(x,y)$ is $x\neq y$ which we already argued
    is pp-definable in $\GR^\ast$. The binary pp-definition for $U_N$
    is simply $\delta_{U_N}:= R(x,y)$, and the binary pp-definition for $U_E$ is
    $\delta_{U_E} := B(x,y)$.
    Finally, the $6$-ary pp-definition for $T$ is
     \[
        \delta_T(x_1,y_1,x_2,y_2,x_3,y_3):= \exists  u,v,w\; \big (S(x_1,y_1,u,v) \land S(x_2,y_2, v,w)\land S(x_3,y_3,w,u) \big ).
     \]  

    First observe that the map $g\mapsto (0,0)(0,1)$, $b\mapsto (0,0)(1,0)$, $r\mapsto (0,1)(1,0)$
    defines a homomorphism from $(A,U_N,U_E,T)$ to $\Pi(\GR^\ast)$. Now, to see that
    $\Pi(\GR^{\ast})\to (A,U_N,U_E,T)$, consider the following map:
    a pair $uv \in N$ is mapped to $r$, a pair $uv \in E$ on a vertical line is mapped to $b$, and a pair $uv \in E$ on a horizontal line is mapped to $g$.
\end{proof}

\begin{proposition}\label{prop:line-multi-bip}
    The sandwich problem for line graphs of bipartite multigraphs is $\NP$-complete.
\end{proposition}
\begin{proof}
     Since $\GR$ is a point-incomparable universal line graph, and the
     class of line graphs of bipartite multigraphs is the class of co-blow-ups of line graphs of
     bipartite graphs, it suffices to show that $\CSP(\GR^\ast)$ is $\NP$-complete (Proposition~\ref{prop:C'-C}).
    In Lemma~\ref{lem:GR-pp-A} we showed that $\GR^\ast$ pp-constructs the finite
    structure $A$, and so it suffices to show that $\CSP(A)$ is $\NP$-complete. 
    This follows from the (easy direction of the) finite-domain dichotomy theorem, Theorem~\ref{thm:dicho},
    if we verify that $A$ has no Siggers polymorphism. This was done by an exhaustive approach with a computer
    program.\footnote{The authors thank Paul Winkler for implementing this.}.
\end{proof}

\section{The Salt-free Vegan Sandwich Problem}
\label{sect:salt-free}
One of the few remaining open cases 
in a sandwich problem classification project from~\cite{alvaradoAOR280} is the complexity of
the sandwich problem for $\{I_4,P_4\}$-free graphs. In this section we prove that the homomorphism
perspective on SP yields a simple answer to this question: the SP for $\{I_4,P_4\}$-free graphs
is $\NP$-complete (Corollary~\ref{cor:perfect-Kk-free}). 

\begin{lemma}\label{lem:bounded-clique}
    Let $\mathcal C$ be a hereditary class of graphs which has the joint embedding property 
    and is preserved under blow-ups. If there is a positive integer $k\ge 3$ such that
    $K_k\in \mathcal C$, and every graph $G\in\mathcal C$ has chromatic number bounded
    by $k$, then the SP for $\mathcal C$ is $\NP$-hard. 
\end{lemma}
\begin{proof}
    Let $H$ be a universal graph in $\mathcal C$. Notice that $H$ is homomorphically
    equivalent to $K_k$: by assumption, $K_k$ embeds into $H$. Conversely, $H$
    maps to $K_k$ by compactness, 
    because every finite subgraph $H'$ of $H$ maps homomorphically to $K_k$.
    Hence, $H^\ast$ pp-constructs $K_k$, and so $\CSP(H^\ast)$ is 
    $\NP$-hard. The claim now follows via Lemma~\ref{lem:H->H*}. 
\end{proof}

If $\mathcal C$ is a class as in Lemma~\ref{lem:bounded-clique} and $G=(V,E)$ 
is a graph, then $(V;E,\varnothing)$ is a yes-instance to the SP for $\mathcal C$
if and only if $\mathcal C$ is $k$-colourable. Actually, this simple reduction 
arising from the CSP approach to SP can be extended to other sandwich problems beyond
the scope of CSPs. Recall that a graph $G$ is  \emph{complete $k$-partite} if its vertex
set admits a partition $(V_1,\dots, V_m)$ with $m\le k$ such that
there is an edge $uv\in E(G)$ if and only if $u\in V_i$ and $v\in V_j$ for $i\neq j$;
equivalently, $G$ is a complete $k$-partite graph if it is a blow-up of a complete
graph on at most $k$ vertices.

\begin{lemma}\label{lem:kcolred}
The SP for a class $\mathcal C$ is $\NP$-hard whenever there is a 
positive integer $k$ such that 
\begin{itemize}
    \item $\chi(G)\le k$ for every $G\in \mathcal C$, and
    \item $\mathcal C$ contains all complete $k$-partite graphs.
\end{itemize}
\end{lemma}
\begin{proof}
    The map $(V,E)\mapsto(V,E,\varnothing)$ is a reduction from $\CSP(K_k)$ to $\SP(\mathcal C)$.
\end{proof}

\begin{theorem}\label{thm:perfect}
    Let $\mathcal F$ be a set of non-complete point-determining graphs such that
    every $\mathcal F$-free graph is a perfect graph. Then for every positive integer
    $k\ge 4$, the sandwich problem for $(\mathcal F\cup\{K_k\})$-free graphs is
    $\NP$-hard.
\end{theorem}
\begin{proof}
    Since every $K_k$-free perfect graph has chromatic number at most
    $k$, it follows that every $(\mathcal F\cup\{K_k\})$-free graph is
    $k$-colourable. Since every graph in $\mathcal F$ is point-determining, 
    every complete $(k-1)$-partite graph is $(\mathcal F\cup\{K_k\})$-free.
    It readily follows that a graph $G= (V,E)$ is $k$-colourable if and only if
    $(V,E,\varnothing)$ is a yes-instance to the SP for $(\mathcal F\cup\{K_k\})$-free
    graphs. The hardness of the problem now follows because $k-1\ge 3$. 
\end{proof}

\begin{corollary}\label{cor:perfect-Kk-free}
    For every positive integer $k$ one of the following statements holds:
    \begin{itemize}
        \item $k\le 3$, and in this case the SP for $\{P_4,K_k\}$-free graphs, the SP for
        $\{P_4,I_k\}$-free graphs, 
        and the SP for $K_k$-free perfect graphs are polynomial-time solvable;
        \item otherwise, $k\ge 4$, and in this case the SP for
        $\{P_4,K_k\}$-free graphs, the SP for
        $\{P_4,I_k\}$-free graphs, and the SP for $K_k$-free perfect graphs are $\NP$-complete.
    \end{itemize}
\end{corollary}
\begin{proof}
    By going to the complement, the statements about the SP for $\{P_4,K_k\}$-free
    graphs imply the statement about the SP for $\{P_4,I_k\}$-free graphs, so we only consider
    the SP for $\{P_4,K_k\}$-free graphs and the SP for $K_k$-free perfect graphs. 
    Containment in NP follows because $K_k$-free perfect graphs and $\{P_4,K_k\}$-free graphs
    can be recognized in polynomial-time (a polynomial-time algorithm to test for perfect graphs follows from~\cite{PerfectAlg,StrongPerfect}). 
    The hardness claim in the second item follows from Theorem~\ref{thm:perfect}, since $P_4$-free graphs
    are easily seen to be perfect.
    The tractable cases are trivial for $k \le 2$. For $k = 3$, 
    the SP for $K_3$-free perfect graphs is just the SP for
    bipartite graphs (which is in P). The SP for $\{P_4,K_3\}$-free
    graphs is the SP for the class of disjoint unions of
    complete bipartite graphs, and thus also polynomial-time solvable.
\end{proof}

\section{The Gy\'arf\'as--Sumner Sandwich Problem}
\label{sect:Gs--SP}

The Gy\'arf\'as--Sumner conjecture (see, e.g.,\cite{CHUDNOVSKY201411}) asserts that for every tree
$T$ there is a function $f\colon\mathbb Z^+\to \mathbb Z^+$ such that
the chromatic number of a $T$-free graph $G$ is bounded by
$f(\omega(G))$, where $\omega(G)$ is the size of a maximum complete subgraph of $G$.

The Brakensiek--Gurusuwami conjecture 
is a hardness conjecture for 
so-called \emph{promise CSPs}. A promise CSP is given by a pair of relational structures $(B,C)$
with a homomorphism from $B$ to $C$. 
The task is to decide whether a given finite structure $A$ has a homomorphism to $B$, or not even a homomorphism to $C$. 
If neither of the two cases applies, the algorithm can answer arbitrarily. This computational problem is denoted by $\PCSP(B,C)$. 
Brakensiek and Gurusuwami~\cite{BrakensiekGuruswami18} conjectured that $\PCSP(K_3,K_k)$ is NP-hard for every $k \geq 3$.
The two conjectures together imply the following one.

\begin{conjecture}[The Gy\'arf\'as--Sumner Sandwich Problem]\label{conj:TKk}
    For every positive integer $k\ge 4$ and every (possibly infinite) set of non-star trees
    $\mathcal T$, 
    the sandwich problem for $(\mathcal T\cup\{K_k\})$-free graphs is $\NP$-hard.
\end{conjecture}
\begin{proof}[Proof of ``GS-conjecture + BG-conjecture $\implies$ Conjecture~\ref{conj:TKk}'']
    By the GS-conjecture there is a positive integer $N$ such that
    every $(\mathcal T\cup \{K_k\})$-free graph is $N$-colourable. Since 
    no tree $T$ in $\mathcal T$ is a star,
    every complete $(k-1)$-partite graph is $(\mathcal T\cup \{K_k\})$-free.
    Therefore, the reduction $(V,E)\mapsto (V,E,\varnothing)$ satisfies the
    following: 
    \begin{itemize}
        \item if $(V,E)$ is $(k-1)$-colourable, then $(V,E,\varnothing)$ is a 
        yes-instance to the SP for $\{T,K_k\}$-free graphs, and
        \item if $(V,E)$ is not $N$-colourable, then $(V,E,\varnothing)$ is
        a no-instance to the SP for $(\mathcal T\cup \{K_k\})$-free graphs. 
    \end{itemize}
    Therefore, the reduction $(V,E)\mapsto (V,E,\varnothing)$ is a
    valid reduction from $\PCSP(K_{k-1},K_N)$ to the SP for
    $(\mathcal T\cup \{K_k\})$-free-graphs. Therefore, the BG-conjecture implies
    that this SP is $\NP$-hard.
\end{proof}

The Gy\'arf\'as--Sumner conjecture has been confirmed in certain cases, e.g., for 
paths~\cite{Gyarfas1987,scott2020survey}. 
The sandwich problem for $\{P_n,K_k\}$-free is a CSP whenever $n \ge 4$, 
and in these cases,
we can also confirm Conjecture~\ref{conj:TKk}. To do so, we first prove the following lemma, which builds on known results
from PCSPs.

\begin{lemma}\label{lem:P5}
    For every integer $k\ge 3$ the SP for $\{P_5,K_k\}$-free 
    graphs is $\NP$-complete.
\end{lemma}
\begin{proof}
    The mapping $(V,E,N)\mapsto (V\cup\{v'\},E\cup\{vv'\colon u\in V\},N)$ is
    a polynomial-time reduction from the SP for $\{K_{k-1},P_5\}$-free graphs to the
    SP for $\{P_5,K_k\}$-free graphs. Hence, it suffices to prove the claim for $k = 3$. 
    Sumner proved that every $\{P_5,K_3\}$-free graph is $3$-colourable~\cite{sumnerTAG},
    and clearly, every blow-up of $C_5$ is $\{P_5,K_3\}$-free. This means that if a graph
    $(V,E)$ admits a homomorphism to $C_5$, then $(V,E,\varnothing)$ is a yes-instance to
    the SP for $\{P_5,K_3\}$-free graphs, and if $(V,E,\varnothing)$ is a yes-instance to
    the SP for $\{P_5,K_3\}$-free graphs, then  $(V,E)$ is $3$-colourable. Hence, $\PCSP(C_5,K_3)$
    reduces in polynomial-time to the SP for $\{P_5,K_3\}$-free graphs, and this proves the
    claim because $\PCSP(C_5,K_3)$ is $\NP$-hard~\cite{BartoBKO21}.
\end{proof}

\begin{theorem}
    For all integers $n,k\ge 4$  the SP for $\{P_n,K_k\}$-free graphs is $\NP$-complete.
\end{theorem}
\begin{proof}
    The case $n = 4$ follows by Corollary~\ref{cor:perfect-Kk-free}, and the case $n = 5$ follows by 
    Lemma~\ref{lem:P5}. Finally, for the
    cases $n\ge 6$ we observe that the SP for $\{P_{n-2},K_k\}$-free graphs reduces in polynomial time
    to the SP for $\{P_n,K_k\}$-free graphs. Indeed, consider an input $(V,E,N)$ to the former
    problem and construct the following input $(V',E',N')$: $V'$ contains $V$
    and for every $v\in V$ a new vertex $u_v$; $E'$ contains $E$ and for every $v\in V$
    we include the edge $vu_v$ in $E'$; and $N'$ contains $N$ together with all pairs
    $u_vu_w$ and $u_vw$ whenever $v \neq w$.
    Observe that the mapping
    $(V,E,N)\mapsto (V',E',N')$ is a polynomial-time
    reduction from the SP for $\{P_{n-2},K_k\}$-free graphs to the SP for
    $\{P_n,K_k\}$-free graphs: for every edge set $E_1$ with $E\subseteq E_1$ 
    and $E_1\cap N'= \varnothing$ there is a clique on $k\ge 4$ vertices in $(V,E_1)$ 
    if and only if there is a clique on $k$ vertices in $(V',E' \cup E_1)$, and
    every induced path in $(V,E_1)$ can be extended by one leaf on each side
    in $(V',E'\cup E_1)$.
\end{proof}

\section{A coNP-intermediate Graph Sandwich Problem}
\label{sect:non-dicho}
Note that for every graph class $\mathcal K$ which is closed under taking (not necessarily induced) subgraphs,
the sandwich problem for ${\mathcal K}$ and the recognition problem for ${\mathcal K}$ clearly have the
same computational complexity (see also~\cite[Proposition~3.1]{golumbicJA19}).

\begin{theorem}\label{thm:no-dicho}
    If $\PO \neq \coNP$, then there exists a graph class ${\mathcal C}$ such that 
    $\SP({\mathcal C})$ is coNP-intermediate. The class 
    ${\mathcal C}$ can be chosen to be closed under subgraphs, 
    and the recognition problem for ${\mathcal C}$ is coNP-intermediate as well. 
\end{theorem}

For a set $S$ of integers $n \geq 5$, let ${\mathcal K}_S$ be the class
of $C_n$-subgraph-free graphs, 
for $n \in S$. 
Clearly, this class is hereditary and even closed under subgraphs, and closed under disjoint unions.

\begin{observation}\label{obs:coNP}
If there is a polynomial-time algorithm that decides whether $n \in S$, for $n \in {\mathbb N}$ given in unary,
then $\SP({\mathcal K}_S)$ is in $\coNP$: on input $(G,G')$,  verify non-deterministically for all subgraphs $H$ of $G$ with $|V(H)|=n$ whether $H$ is isomorphic to $C_n$ (which can be done in polynomial time).
\end{observation}

\begin{observation}
    \label{obs:reduce}
    If $S' \subseteq S$ and $S \setminus S'$ is finite, then 
    there is a polynomial-time reduction from 
    $\SP({\mathcal K}_{S})$ to  $\SP({\mathcal K}_{S'})$. 
In particular, if $S$ is finite, then $\SP({\mathcal K}_S)$ is in P (since $\SP({\mathcal K}_{\emptyset})$ is clearly in P). 
\end{observation}

We first construct a graph class ${\mathcal K}$ whose sandwich problem is coNP-complete; this will be needed in the proof of 
Theorem~\ref{thm:no-dicho}. Define $k_0,k_1,k_2,\dots$
inductively as follows:
$k_0 :=5$, and $k_{i+1} := 3 k_i$ for all $i \in {\mathbb N}$. 
Let $T := \{k_i \mid i \in {\mathbb N}\}$. 
Note that for a number $n \in {\mathbb N}$ given in unary, one can decide in polynomial
time in $n$ whether $n \in T$. Also note that all numbers in $T$ are odd.

\begin{proposition}
    The graph sandwich problem for ${\mathcal K}_T$ is $\coNP$-complete.
\end{proposition}
\begin{proof}
    By Observation~\ref{obs:coNP}, $\SP({\mathcal K}_T)$ is in coNP. 
    We present a polynomial-time reduction from the complement of the following computational problem, which is well-known to be NP-complete (see, e.g.,~\cite{doi:10.1137/0211056}), to $\SP({\mathcal K}_T)$.
    The input of this problem consists of a finite bipartite graph $G$ with special vertices $s$ and $t$, and the task is to decide whether there exists a Hamiltonian path from $s$ to $t$.  
    
    Let $n := |V(G)|$.
    Let $i \in {\mathbb N}$ be the unique number
    such that $k_i \leq n < k_{i+1}$. It is easy to see that $i$ and $k_i$ can be
    computed in polynomial time. Let $H$ be the graph obtained from $G$ by adding a path $P_l$ with $l \geq 1$ vertices, and adding an edge from $t$ to the first vertex of the path, and an edge from the last vertex of the path to $s$. We choose $l$ such that
the resulting graph $H$ has exactly $k_{i+2} = 3 k_{i+1}$ vertices.
We claim that $G$ has a Hamiltonian path from $s$ to $t$ if and only if $H$ has a cycle of length $s$, for some $s \in T$. 

If $G$ has a Hamiltonian path $u_0,\dots,u_{n-1}$, then 
we may use the copy of $P_l$ in $H$ to find a Hamiltonian cycle in $H$, which has length $|V(H)| = k_{i+2} \in T$.
Conversely, suppose that $H$ has a cycle
of length $k \in T$. Since $k \in T$ is odd, and
$G$ is bipartite, the cycle must traverse all of $P_l$. 
Since $n \leq p_{i+1} < l$, we must have $k = p_{i+2}$.  
Thus, $H$ has a 
Hamiltonian cycle. We may assume that the cycle starts in $s$ and that it first traverses $G$ and then $P_l$.
Hence, we found a Hamiltonian path from $s$ to $t$ in $G$, and thus $H \notin \SP({\mathcal K}_T)$.
\end{proof}

\begin{proof}[Proof of Theorem~\ref{thm:no-dicho}]
We make a Ladner-type construction~\cite{Ladner,BodirskyGrohe}, and define a set $S \subseteq {\mathbb N}$ such that $\SP({\mathcal K}_S)$ is coNP-intermediate using a Turing machine $F$ as follows. We fix one of the standard encodings of undirected graphs as strings in order to perform computations on graphs with Turing machines.
Let $M_1,M_2,\dots$ be an enumeration of all polynomial-time bounded Turing machines, and let $R_1,R_2,\dots$ be an enumeration of all polynomial-time bounded reductions. We assume that these enumerations are effective; it is well known that such enumerations exist. 

Let $F$ be the Turing machine which takes as input one number $n$ in unary representation and returns another number $f(n) \in {\mathbb N}$; the set $S$ will then be  
$$ S := \{k_n \mid f(n) \text{ is even} \}$$
for $k_0,k_1,k_2,\dots$ as defined above. 

The machine simulates itself on input $1$, then on input $2$, and so on, until the number of computation steps exceeds $n$. 
Let $k$ be the value computed by
$F$ for the last input $i$ on which the simulation was completely performed. 

Next, $F$ enumerates for $s=1,2,3,\dots$ all undirected graphs $G$ on the vertex set $\{1,\dots,s\}$.
If $k$ is even, then for each graph $G$ in the enumeration, $F$ simulates $M_{k/2}$ on the encoding of $G$. 
Moreover, $F$ computes whether $G$ belongs to ${\mathcal K}_S$, again simulating itself. 
If 
\begin{itemize}
    \item $M_{k/2}$ rejects and $G \in {\mathcal K}_S$, or 
    \item $M_{k/2}$ accepts and 
    $G \notin {\mathcal K}_S$,
\end{itemize}
then $F$ returns $k+1$. 

If $k$ is odd, then  
$F$ simulates the computation of $R_{\lfloor k/2 \rfloor}$ on the encoding of $G$. Then $F$ computes whether the output of
$R_{\lfloor k/2 \rfloor}$ encodes a graph $G' \in {\mathcal K}_T$, 
and whether $G \in {\mathcal K}_S$, again simulating itself. 
If 
\begin{itemize}
    \item $G \in {\mathcal K}_S$ and $G' \notin {\mathcal K}_T$, or 
    \item $G \notin {\mathcal K}_S$ and $G' \in {\mathcal K}_T$, 
\end{itemize}
then $F$ returns $k+1$. 
Whenever the number of computation steps of $F$ exceeds $2n$ the machine returns $k$.

{\bf Claim 1.} 
$\SP({\mathcal K}_S)$ is in coNP. This follows from Observation~\ref{obs:coNP}, because on input $n$ the Turing machine $F$ takes at most $2n$ computation steps. 

{\bf Claim 2.} $\SP({\mathcal K}_S)$ is not in P. 
If $\SP({\mathcal K}_S)$ is in P, then there exists an $i \in {\mathbb N}$ such that $M_i$ decides 
$\SP({\mathcal K}_S)$. 
It can then be shown that $f(n)$ is even for all but finitely many $n$; thus, $T \setminus S$ is finite, and hence
$\SP({\mathcal K}_T)$ reduces to
$\SP({\mathcal K}_S)$ by Observation~\ref{obs:reduce}, 
in contradiction to our assumption that $P \neq \coNP$.

{\bf Claim 3.} 
$\SP({\mathcal K}_S)$ is not $\coNP$-hard. 
If $\SP({\mathcal K}_S)$ is coNP-hard, then there exists an $i \in {\mathbb N}$ such that $R_i$ is a reduction from $\SP({\mathcal K}_T)$ to $\SP({\mathcal K}_S)$. 
In this case it can be shown that $f(n)$ is odd for all but finitely many $n$. This implies that
$S$ is finite, and hence 
$\SP({\mathcal K}_S)$ is in P by Observation~\ref{obs:reduce}. 
Again, we reach a contradiction to our assumption that  $P \neq \coNP$.
\end{proof}

\section{Conclusions and Open Problems}
\label{sec:conclusions}

We have established a connection between graph sandwich problems and (infinite-domain) constraint satisfaction problems:
we clarified under which conditions on a graph class ${\mathcal C}$ the graph sandwich problem for ${\mathcal C}$
is of the form $\CSP(H^*)$ for some countably infinite graph $H$. The theory of CSPs is particularly well-developed for
\emph{$\omega$-categorical} structures $H^*$ (introduced below; every finite structure is $\omega$-categorical), 
because for such structures, one can use concepts and results from universal algebra, which was crucial for solving the
Feder-Vardi dichotomy conjecture mentioned in the introduction.

An \emph{automorphism} of a $\tau$-structure $A$ is a bijection
$f\colon A\to A$ such that for each $R\in \tau$ of arity $r$
there is a tuple $(a_1,\dots, a_r)\in  R(A)$ if and only
if $(f(a_1),\dots, f(a_r))\in A$. We denote by $\Aut(A)$ the
automorphism group of $A$. 
A structure $A$ is  \emph{$\omega$-categorical}
if $\Aut(A)$ has for every positive integer $k$ finitely many orbits in its componentwise
action on $k$-tuples; this is not the official definition, but equivalent to it by the
theorem of Engeler, Svenonius, and Ryll-Nardzewski (see, e.g.,~\cite{Hodges}). 
Every finite structure is $\omega$-categorical, because there are only 
finitely many $k$-tuples for every positive integer $k$. 
The infinite clique $(\mathbb N; \{(x,y)\in \mathbb N^2\colon x\neq y\})$
is $\omega$-categorical because $(x_1,\dots, x_k)$ and
$(y_1,\dots, y_k)$ belongs to the same orbit whenever
$x_i = x_j\iff y_i = y_j$ for all $i,j\in[k]$. 
It is easy to see that every structure which is first-order interpretable
in an $\omega$-categorical structure is $\omega$-categorical as well, which shows that 
the universal graphs from  Sections \ref{sect:multipartite},
\ref{sect:threshold},  and~\ref{sect:permutation} 
are $\omega$-categorical.  In fact, it can be shown that the countably infinite
universal structures that appeared in  Sections \ref{sect:line}, \ref{sect:split},
and~\ref{sect:comparability} are $\omega$-categorical as well.

On the other hand, 
the infinite directed path $\vec{P}:= (\mathbb Z,\{(x,y)\colon y = x +1\})$
is not $\omega$-categorical: for every pair of positive integers $n\neq m$,
there is no automorphism of $\vec{P}$ mapping $(0,n)$ to $(0,m)$, and so
there are infinitely many orbits of pairs. 

Automorphisms preserve first-order
formulas (i.e., if $A\models\phi(a_1,\dots, a_n)$ and $f\in \Aut(A)$, 
then $A\models \phi(f(a_1),\dots, f(a_n))$), and so, if
$B$ is a first-order reduct of $A$, then every automorphism
of $A$ is an automorphism of $B$. In other words, 
$\Aut(A)$ is a subgroup of $\Aut(B)$ as a permutation group. 
It turns out that when $A$ and $B$ are $\omega$-categorical 
then the converse also holds: $B$ is a first-order reduct 
of $A$ if and only if $\Aut(A) \subseteq \Aut(B)$~\cite[Corollary 7.3.3]{HodgesLong}.
Similarly, $B$ is a primitive positive reduct of $A$ (i.e., all relations of
$B$ are primitively positively definable in $A$) if and only if $\Pol(A) \subseteq \Pol(B)$. 
It follows that the complexity of the CSP of an $\omega$-categorical structure $B$
only depends on $\Pol(B)$.

 \begin{corollary}\label{cor:SP-CSP}
     Let $\calC$ be a hereditary class with the joint embedding property, and
     preserved by blow-ups or by co-blow-ups. If there is a universal $\omega$-categorical
     graph in $\calC$, then $\SP(\calC)$ equals $\CSP(H^*)$ for an
     $\omega$-categorical graph $H$, and its complexity only depends on
     $\Pol(H^*)$ (see Section~\ref{sect:threshold}).
 \end{corollary}

This motivates the following research questions which are left for the future. 

\begin{itemize}
    \item Is there an $\omega$-categorical perfect graph? 
    \item Is there a finite set of finite graphs ${\mathcal F}$ such that the SP for ${\mathcal F}$-free graphs is \emph{NP-intermediate}, i.e., in NP, but neither in P nor NP-complete? 
    \item Is there an algorithm that decides for a given finite set of finite graphs whether the graph sandwich problem for the class $\mathcal C$ of ${\mathcal F}$-free graphs is a CSP? Braunfeld~\cite{BraunfeldUndec} showed that given a finite set of finite graphs ${\mathcal F}$, it is undecidable whether the class of ${\mathcal F}$-free graphs has the JEP; hence, by our results, it suffices to verify that the undecidability result of Braundfeld works when restricted to point determining graphs. 
    \item Is there an algorithm that decides for a given finite set of connected point-determining graphs $\mathcal F$
    whether the SP for $\mathcal F$-free graphs is in P (assuming P$\neq$NP)?
    This problem is sometimes referred to as the \emph{tractability problem} which is known to be undecidable for
    several formalisms (e.g., for existential second-order logic~\cite{Book}).
\end{itemize}

\bibliographystyle{abbrv}
\bibliography{global.bib}

\end{document}